%% file: rita.tex
\documentclass{sig-alternate}

\usepackage{times}
\usepackage{verbatim}
\usepackage{amsmath}
\usepackage[normalem]{ulem}
\usepackage[linesnumbered,vlined]{algorithm2e}
\usepackage{latexsym}
\usepackage[textwidth=2cm]{todonotes}
\usepackage{algorithmic}
\usepackage{graphicx}
\usepackage{wrapfig}
\usepackage{enumitem}
\usepackage{multirow}
\usepackage{url}
\usepackage{empheq}
\usepackage{mathptmx}
\setlist{leftmargin=.8em, itemsep=1pt, topsep=1pt}
\usepackage{cmm-greek} 
\begin{document}


\numberofauthors{5}
\author{
	Quoc Trung Tran  \\
  \affaddr{UC Santa Cruz} \\
  \affaddr{tqtrung@cs.ucsc.edu}
  \and 
  Ivo Jimenez \\
  \affaddr{UC Santa Cruz} \\
  \affaddr{ivo@cs.ucsc.edu}
  \and
  Rui Wang \\
  \affaddr{UC Santa Cruz} \\
  \affaddr{rwang10@cs.ucsc.edu}
  \and
  Neoklis Polyzotis \\
  \affaddr{UC Santa Cruz}\\
  \affaddr{alkis@cs.ucsc.edu}
  \and Anastasia Ailamaki \\
  \affaddr{\'{E}cole Polytechnique \'{F}ed\'{e}rale de Lausanne}\\
  \affaddr{natassa@epfl.ch}
}

\input{term}

\title{RITA: An Index-Tuning Advisor for Replicated Databases}

%
%

\maketitle

\input{abstract}

\input{intro}
\input{related}

\input{problem}

\input{divergent}
\input{feature}

\input{expt}

\input{conclusion}


%
{\small
\bibliographystyle{abbrv}
\bibliography{ref}  
}
%
%
\balancecolumns

\newpage
\input{appendix}

\balancecolumns 
\end{document}

%% file: term.tex

\newcommand{\stitle}[1]{\vspace{0.8ex} \noindent{\bf #1}}
\newcommand{\eop}{\hspace*{\fill}\mbox{$\Box$}}

\newtheorem{definition}{Definition}
\newtheorem{lemma}{Lemma}
\newtheorem{problem}{Problem}
\newtheorem{theorem}{Theorem}
\newtheorem{property}{Property}
\newtheorem{example}{Example}
\newtheorem{corollary}{Corollary}

\newcommand{\pupd}{\mathit{p_{upd}}}

\newcommand{\noindex}[1]{\mbox{\small\rm  SCAN}_{\small #1}}
\newcommand{\noindextiny}[1]{\mbox{\tiny\rm SCAN}_{\tiny #1}}

\newcommand{\eat}[1]{}
\newcommand{\Iset}{{\cal S}}
\newcommand{\Xsol}{{\cal X}^{*}}
\newcommand{\Xsup}{{\cal X}^{core}}

\newcommand{\contribution}[1]{{\bf (#1.)}}

\newcommand{\Iopt}{I_{opt}}
\newcommand{\Ibalance}{I_{balance}}
\newcommand{\workloadonline}{W_{online}}
\newcommand{\Icurrent}{I^{c}}

\newcommand{\va}{{\bf v}}
\newcommand{\vastar}{{\bf v^{*}}}
\newcommand{\vacounter}{{\bf v^{c}}}

\newcommand{\rita}{{\advisor}}
\newcommand{\cost}{{\mathit{cost}}}
\newcommand{\totalcost}{\mathit{TotalCost}}
\newcommand{\ftotalcost}{\mathit{FTotalCost}}

\newcommand{\TotalCost}{\totalcost} 
\newcommand{\querycost}{\mathit{QueryCost}}
\newcommand{\updatecost}{\mathit{UpdateCost}}

\newcommand{\mapping}{\mathbf{h}}
\newcommand{\routing}{\mapping}
\newcommand{\configuration}{\mathbf{I}}
\newcommand{\rcomponent}[1]{\mathit{h}_{#1}}
\newcommand{\route}[1]{\rcomponent{#1}}
\newcommand{\confcomponent}[1]{\mathit{I}_{#1}}
\newcommand{\expected}{\mathit{ExpTotalCost}}
\newcommand{\Cmat}{{\mathit{C_m}}}

\newcommand{\ccost}{\mathit{ccost}}
\newcommand{\dcost}{\mathit{dcost}}
\newcommand{\rep}{r}
\newcommand{\costopt}{\mathit{cost^{opt}}}
\newcommand{\yopt}{\mathit{yo}}
\newcommand{\xopt}{\mathit{xo}}

\newcommand{\Ccand}{{\it C_{cand}}}

\newcommand{\deploy}{z}
\newcommand{\Ndeploy}{N_{d}}
\newcommand{\Cdeploy}{C_{d}}

\newcommand{\divergent}{{\it DDT}}
\newcommand{\dit}{{\it DDT}}
\newcommand{\workload}{W}

\newcommand{\divgdesign}{\mbox{\textsc{DivgDesign}}}
\newcommand{\divgbip}{\mbox{\textsc{DivBIP}}}
\newcommand{\advisor}{\mbox{\textsc{RITA}}}
\newcommand{\divgbipgreedy}{\mbox{\textsc{DivBIP-greedy}}}
\newcommand{\NAME}{\divgbip}
\newcommand{\name}{\divgbip}
\newcommand{\DIVGDESIGN}{\divgdesign}
\newcommand{\divgunif}{\mbox{\textsc{Unif}}}
\newcommand{\unif}{\mbox{\textsc{Unif}}}
\newcommand{\divgmethod}{\mbox{\textsc{BIP}}}

\newcommand{\load}{\mathit{load}}
\newcommand{\newload}{\mathit{\bar{load}}}
\newcommand{\deltaload}{\ensuremath{\Delta\mathit{load}}}
\newcommand{\loadfail}{\mathit{fload}}

\newcommand{\singlecost}{\mathit{SingleCost}}

\newcommand{\DA}{\mathit{DBAdv}}
\newcommand{\Cx}[1]{C^{#1}}
\newcommand{\divconfiguration}{{\cal I}}

\newcommand{\failurefactor}{{\alpha}}
\newcommand{\nodefactor}{{\tau}}

\newcommand{\punif}{{p_{\mbox{\tiny\rm UNIF}}}}
\newcommand{\pdivg}{{p_{\mbox{\tiny\rm RITA}}}}
\newcommand{\prita}{{p_{\mbox{\tiny\rm RITA}}}}
\newcommand{\pdd}{{p_{\mbox{\tiny\rm DD}}}}
\newcommand{\pcurr}{p_{\mbox{\tiny\rm RITA}}^{\mbox{\tiny\rm curr}}}
\newcommand{\pslide}{p_{\mbox{\tiny\rm RITA}}^{\mbox{\tiny\rm slide}}}
\newcommand{\slide}{{\mbox{\tiny\rm slide}}}
\newcommand{\curr}{{\mbox{\tiny\rm curr}}}

\newcommand{\tplans}{\mathit{TPlans}}
\newcommand{\TPlans}[1]{\mathit{TPlans}(#1)}
\newcommand{\atom}{\mathit{Atom}}

\newcommand{\static}{\mathit{TPCDS\mbox{-}query}}
\newcommand{\mix}{\mathit{TPCDS\mbox{-}mix}}
\newcommand{\dynamic}{\mathit{TPCDS\mbox{-}dyn}}

\newcommand{\costbip}{{\hat{\cost}}}
\newcommand{\costbipf}{{\hat{\fcost}}}
\newcommand{\biptotalcost}{\hat{\totalcost}}
\newcommand{\bipquerycost}{\hat{\querycost}}
\newcommand{\bipupdatecost}{\hat{\updatecost}}

\newenvironment{myalgorithmic}{%
\renewenvironment{algocf}[1][h]{}{}
\algorithm
\small
}{%
\endalgorithm
}
\newcommand{\FunctionTitle}[1]{
  \everypar={\relax}
  \setcounter{AlgoLine}{0} 
  \textbf{Function }#1 \\
}
\newcommand{\ProcedureTitle}[1]{
  \everypar={\relax}
  \setcounter{AlgoLine}{0} 
  \textbf{Procedure }#1 \\
}
\newcommand{\FuncSep}{
  \BlankLine
  \BlankLine
  \BlankLine
}
\SetKwInput{KwInit}{Initialization}
\SetKwInput{KwKnobs}{Knobs}
\newcommand{\dc}{\ensuremath{\!\!\dblcolon\!\!}}
\newcommand{\smallCommentText}[1]{{\small\it #1}}
\SetCommentSty{smallCommentText}

\newcommand\tightpara{\looseness=-1}

\newif\ifshowadd
\newif\ifshowrem


\newcommand\edit[2]{%
	\ifshowadd%
		\textcolor{red}{#1}%
	\else%
		{#1}%
	\fi%
    \ifshowrem%
        { \textrm{\sout{#2}}}%
    \fi%
}

\newcommand\rem[1]{%
	\ifshowrem%
		\textrm{\sout{#1}}%
	\fi%
}

\newcommand\add[1]{%
	\ifshowadd%
		\textcolor{red}{#1}%
	\else%
		{#1}%
	\fi%
}

\newcommand\cspace{\vspace*{-1.2em}}

\newlength\dlf
\newcommand\alignedbox[2]{
  &
  \begingroup
  \settowidth\dlf{$\displaystyle #1$}
  \addtolength\dlf{\fboxsep+\fboxrule}
  \hspace{-\dlf}
  \boxed{#1 #2}
  \endgroup
}

%% file: abstract.tex
\begin{abstract}
Given a replicated database, a divergent design tunes the indexes 
in each replica differently in order to specialize it 
for a specific subset of the workload. 
This specialization brings significant performance gains 
compared to the common practice of having the same indexes in all replicas, 
but requires the development of new tuning tools for database administrators. 
In this paper we introduce $\rita$ 
({\bf R}eplication-aware {\bf I}ndex {\bf T}uning {\bf A}dvisor), 
a novel divergent-tuning advisor that offers several essential features 
not found in existing tools: 
\add{
it generates robust divergent
designs that allow the system to adapt gracefully to replica failures;
}
it computes designs that spread the load evenly among specialized replicas, 
both during normal operation and when replicas fail; it monitors the workload online in order to detect changes that require a recomputation of the divergent design; and, it offers suggestions to elastically reconfigure the system (by adding/removing replicas or adding/dropping indexes) to respond to workload changes. The key technical innovation behind $\rita$ is showing that the problem of selecting an optimal design can be formulated as a Binary Integer Program (BIP).
The BIP has a relatively small number of variables, 
which makes it feasible to solve it efficiently using any off-the-shelf 
linear-optimization software. 
Experimental
results demonstrate that $\rita$ 
computes better divergent designs compared to existing tools,
offers more features, and has fast execution times. 

\end{abstract}

%% file: intro.tex

\section{Introduction}
\label{sec:intro}

Database replication is used heavily in distributed systems and database-as-a-service platforms (e.g., Amazon's Relational Database Service~\cite{RDS} or Microsoft SQL Azure~\cite{azure}), to increase availability and to improve performance through parallel processing. The database is typically replicated across several nodes, and replicas are kept synchronized (eagerly or lazily) when updates occur so that incoming queries can be evaluated on any replica. 

\emph{Divergent-design tuning}~\cite{Divergent2012} 
represents a new paradigm to tune workload performance over a replicated database.
A divergent design leverages replication as follows: it specializes each replica to a specific subset of the workload by installing indexes that are particularly beneficial for the corresponding workload statements. Thus, queries can be evaluated more efficiently by being routed to a specialized replica. As shown in a previous study, a divergent design brings significant performance improvements when compared to a uniform design that uses the same indexes in all replicas: queries are executed faster due to replica specialization (up to 2x improvement on standard benchmarks), but updates as well become significantly more efficient (more than 2x improvement) since fewer indexes need to be installed per replica. 

To reap the benefits of divergent designs in practice, DB administrators need new index-tuning advisors that are replication-aware. 
The original study~\cite{Divergent2012} introduces
an advisor called $\divgdesign$, which creates specialized designs per replica but has severe limitations that restrict its usefulness in practice. Firstly,
$\divgdesign$ assumes that replicas are always operational. Replica failures, however, are common in real systems, and the resulting workload redistribution may cause queries to be routed to low-performing replicas, with predictably negative effects on the overall system performance. An effective advisor should generate robust divergent designs that allow the system to adapt gracefully to replica failures.
Secondly, $\divgdesign$ ignores the effect of specialization to each replica's load, and can therefore incur a skewed load distribution in the system. Our experiments suggest that  $\divgdesign$ can cause certain replicas to be twice as loaded as others. A good advisor should take the replica load into account, and generate divergent designs that provide the benefits of specialization while maintaining a balanced load distribution.
Lastly, $\divgdesign$ targets a static system where the database workload and the number of replicas are assumed to remain unchanged. 
A replicated database system, however, is typically volatile: the workload may change over time, and in response the DBA may wish to elastically reconfigure the system by expanding or shrinking the set of replicas and by incrementally adding or dropping indexes at different replicas. 
A replication-aware advisor should alert the DBA when a workload change necessitates 
retuning the divergent design, and also help the DBA evaluate options for changing the design.

The limitations of $\divgdesign$ stem from the fact that 
it internally employs a conventional index-tuning advisor, e.g., DB2's {\tt db2advis} 
or the index advisor of MS SQL Server, which is not suitable for 
modeling and solving the aforementioned issues. 
Modifying $\divgdesign$ to address its limitations 
would require a non-trivial redesign of the advisor. 
A more general question is whether it is even feasible to 
reap the performance benefits demonstrated in~\cite{Divergent2012} 
and at the same time maintain a balanced load and 
the ability to adapt gracefully to failures. 
Our work shows that this is indeed feasible 
but requires the development of a new type of index-tuning advisor that is replication-aware.

\stitle{Contributions.} In this paper, we introduce a novel index advisor termed $\rita$ (Replication-aware Index Tuning Advisor) that provides DBAs with a powerful tool for divergent index tuning. Instead of relying on conventional techniques for index tuning, $\rita$ is a new type of index advisor that is designed from the ground up 
to take into account replication and the unique characteristics of divergent designs.
$\rita$'s foundation is a novel reduction of the problem of divergent design tuning to Binary Integer Programming (BIP). The BIP formulation allows $\rita$ to employ an off-the-shelf linear optimization solver to compute near-optimal designs that satisfy complex constraints (e.g., even load distribution or robustness to failures). Compared to $\divgdesign$, $\rita$ offers richer tuning functionality and is able to compute divergent designs that result in significantly better performance. 

More concretely, the contributions of our work can be summarized as follows:

\noindent
\contribution{1} To make divergent designs suitable for the characteristics of real-world systems, we introduce a generalized version of the problem of divergent design tuning that has two important features: it takes into account the probability of replica failures and their effect on workload performance; and, it allows for an expanded class of constraints on the computed divergent design and in particular constraints on global system-properties, e.g., maintaining an even load distribution (Section~\ref{sec:problem}).

\noindent
\contribution{2} We prove that, under realistic assumptions about the underlying system, the generalized tuning problem can be formulated as a compact Binary Integer Program (BIP), i.e., a linear-optimization problem with a relatively small number of binary variables. The implication is that we can use an off-the-shelf solver to efficiently compute a (near-)optimal divergent design that also satisfies any given constraints (Section~\ref{sec:div-bip}). 
	
\noindent
\contribution{3} We propose $\rita$ as a new index-tuning tool that leverages the previous theoretical result to implement a unique set of features. $\rita$ allows the DBA to initially tune the divergent design of the system using a training workload. Subsequently, $\rita$ continuously analyzes the incoming workload and alerts the DBA if a retuning of the divergent design could lead to substantial performance improvements. The DBA can then examine how to elastically adapt the divergent design to the changed workload, e.g., by expanding/shrinking the set of replicas, incrementally adding/removing indexes, or changing how queries are distributed across replicas. Internally, $\rita$ translates the DBA's requests to BIPs that are solved efficiently by a linear-optimization solver. In fact, $\rita$ often returns its answers in seconds, thus facilitating an exploratory approach to index tuning (Section~\ref{sec:tool}).
	
\noindent
\contribution{4} We perform an extensive experimental study to validate the effectiveness of $\rita$ as a tuning advisor. The results show that the designs computed by $\rita$ can improve system performance by up to \add{a factor of four} compared to the standard uniform design that places the same indexes on all replicas. Moreover, $\rita$ outperforms $\divgdesign$ by up to
\add{a factor of three}
in terms of the performance of the computed divergent designs, while supporting a larger class of constraints (Section~\ref{sec:expt}). 

Overall, $\rita$ provides a positive answer to our previously stated question: a divergent design can bring significant performance benefits 
while maintaining important properties such as 
a balanced load distribution and tolerance to failures. Consequently, divergent design advisors can be practically employed on real systems and guide further development of tuning tools. The underlying theoretical results (problem definition and BIP formulation) are also significant, as they expand on the previous work on single-system tuning~\cite{Dash2011} 
and demonstrate a wider applicability of Binary Integer Programming to index-tuning problems.

\eat{
Database replication is used heavily in distributed database systems and database-as-a-service platforms (e.g., Amazon's Relational Database Service~\cite{RDS} or Microsoft SQL Azure~\cite{azure}) in order to achieve high availability and also to improve performance through parallel processing. In this setting, the database is replicated across several nodes, replicas are kept synchronized (eagerly or lazily) when updates occur, and incoming queries can be evaluated on any replica. 

\emph{Divergent designs}~\cite{Divergent2012} represent a new paradigm to tune workload performance over a replicated database.
A divergent design leverages replication as follows: it specializes each replica to a specific subset of the workload by installing indexes that are particularly beneficial for the corresponding workload statements. Thus, queries can be evaluated more efficiently by being routed to a specialized replica. As shown in a previous study, a divergent design brings significant performance improvements compared to a uniform design that uses the same indexes in all replicas: queries are executed faster due to replica specialization (up to 2x improvement on standard benchmarks), but updates as well become significantly more efficient (more than 2x improvement) since fewer indexes need to be installed per replica.

To reap the benefits of divergent designs in practice, DB administrators need new index-tuning advisors that are replication-aware. 
The original study~\cite{Divergent2012} introduced
an  advisor called $\divgdesign$, 
but we argue that this advisor has severe limitations that restrict its usefulness in practice.
Specifically, $\divgdesign$ specializes replicas under the assumption that they are always operational. However, replica failures are common in real systems, and the resulting workload redistribution may cause queries to be routed to badly-performing replicas, with predictably bad effects on system performance. An effective advisor should take this possibility into account and generate divergent designs that allow the system to adapt gracefully to replica failures.
Another shortcoming is that $\divgdesign$ ignores how specialization affects the load of each replica, and can thus lead to a skewed load distribution in the system. In fact, our experiments suggest that  $\divgdesign$ can cause certain replicas to be twice as loaded as other replicas. A good advisor should take replica-load into account, and generate divergent designs that provide the benefits of specialization while maintaining a balanced load distribution.
Lastly, $\divgdesign$ essentially targets a static system where the database workload and the number of replicas are assumed to stay unchanged. However, a replicated database system can be quite volatile: the workload may change over time, and in response the DBA may wish to elastically reconfigure the system by expanding/shrinking the set of replicas and by incrementally adding/dropping indexes at different replicas. 
A replication-aware advisor should alert the DBA when a workload change necessitates a retuning of the divergent design, and also help the DBA evaluate options for changing the design.

The limitations of $\divgdesign$ stem from the fact that it internally employs a conventional index-tuning advisor, e.g., DB2's {\tt db2advis} or the index advisor of MS SQL Server, which is not suitable to model and solve the aforementioned issues related to divergent designs. Modifying $\divgdesign$ to address its limitations would require a complete redesign of the advisor which is far from trivial. Moreover, one can ask this question about the feasibility of divergent designs in practice: can we reap the performance benefits demonstrated in~\cite{Divergent2012} when we impose constraints on load-balancing and failure adaptation, or is there an inherent tradeoff between the ability to specialize replicas and these important constraints? Answering this question is key for the further development of tools for divergent design tuning.

\stitle{Our Contributions.} We introduce in this paper a novel index advisor termed $\rita$ (short for Replication-aware Index Tuning Advisor) that provides DBAs with a powerful tool for divergent index tuning. Instead of relying on conventional techniques for index tuning, $\rita$ is a new type of index advisor that is designed from the ground up for divergent designs. $\rita$'s foundation is a novel reduction of the problem of divergent design tuning to Binary Integer Programming (BIP). This  allows $\rita$ to employ an off-the-shelf linear optimization solver in order to compute near-optimal designs that satisfy complex constraints (e.g., even load distribution or robustness to failures). Compared to $\divgdesign$, \rem{the current state-of-the-art,} $\rita$ offers richer tuning functionality and is able to compute divergent designs that result in significantly better performance. 

More concretely, the contributions of our work can be summarized as follows:
\begin{itemize}
	\item (Section~\ref{sec:problem}) To make divergent designs suitable for the characteristics of real-world systems, we introduce a generalized version of the problem of divergent design tuning that has two important features: it takes into account the probability of replica failures and their effect on workload performance; and, it allows for an expanded class of constraints on the computed divergent design and in particular constraints on global system-properties, e.g., maintaining an even load distribution.

	\item (Section~\ref{sec:div-bip}) We prove that, under realistic assumptions about the underlying system, the generalized tuning problem can be formulated as a compact Binary Integer Program (BIP), i.e., a linear-optimization problem with a relatively small number of binary variables. The implication is that we can use an off-the-shelf solver to efficiently compute a (near-)optimal divergent design that also satisfies any given constraints. 
	
	\item (Section~\ref{sec:tool}) We propose $\rita$ as a new index-tuning tool that leverages the previous theoretical result to implement a unique set of features. $\rita$ allows the DBA to initially tune the divergent design of the system using a training workload. Subsequently, $\rita$ continuously analyzes the incoming workload and alerts the DBA if a retuning of the divergent design could lead to substantial performance improvements. The DBA can then examine how to elastically adapt the divergent design to the changed workload, e.g., by expanding/shrinking the set of replicas, incrementally adding/removing indexes, or changing how queries are distributed across replicas. Internally, $\rita$ translates the DBA's requests to BIPs that are solved efficiently by a linear-optimization solver. In fact, $\rita$ often returns its answers in seconds, thus facilitating an exploratory approach to index tuning. 
	
	\item (Section~\ref{sec:expt}) We perform an extensive experimental study to validate the effectiveness of $\rita$ as a tuning advisor. The results show that the designs computed by $\rita$ can improve system performance by up to a factor of two compared to the standard uniform design that places the same indexes on all replicas. Moreover, $\rita$ outperforms $\divgdesign$ by up to 35\% in terms of the performance of the computed divergent designs, while supporting a larger class of constraints.
\end{itemize}

Overall, this paper provides a positive answer to our previous question: a divergent design can bring significant performance benefits \emph{and} maintain important properties such as a balanced load distribution and tolerance to failures. This answer advocates the applicability of divergent designs in real systems and the further development of tuning tools. The underlying theoretical result is also significant, as it expands on the results of previous works on single-system tuning~\cite{Dash2011} and demonstrates a wider applicability of Binary Integer Programming to index-tuning problems.
}

%% file: related.tex
\section{Related work}
\label{sec:related}

\stitle{Index tuning. } 
There has been a long line of research studies on the problem of tuning 
the index configuration of a single DBMS (e.g.,~\cite{Bruno2008, Dash2011,Zrllsgf2004}). 
These methods 
analyze a representative workload  and recommend an index configuration
 that optimizes the evaluation of the workload according to the optimizer's 
estimates. 
A recent study~\cite{Dash2011} has introduced
the COPHY index advisor that
outperforms state-of-the-art commercial and research techniques
by up to an order of magnitude 
in terms of both solution quality and total execution time. 
Both \rita\ and COPHY leverage the same underlying principle of linear composability, which we will define and discuss extensively in Section~\ref{sec:inum}, 
in order to cast the index-tuning problem as a compact, 
efficiently-solvable Binary Integer Program (BIP). 
However, COPHY targets the conventional index-tuning problem 
where the goal is to compute a single index configuration for 
a {\it single-node} system. 
This problem scenario is much simpler than what we consider in our work, where there are several nodes in the system, each can carry a different index configuration, queries have to be distributed in a balanced fashion and the system must recover gracefully from failures. Leveraging the principle of linear composability in this generalized problem scenario is one of the key contributions of our work.

\rem{
Shinobi~\cite{Wu2011} is a system that utilizes workload information 
to partition the data and selectively index data within each partition. 
This results in less expensive index maintenance and reorganization costs, 
by creating and dropping indexes on subsets of the data (the workload-based partitions) 
as the access patterns change.  
However, Shinobi does not address replication of partitions or different index configurations on replicas of 
the same partition, 
which is the problem that we examine in our work. 
In fact, our techniques can be used to determine which indexes to install on each replica, 
and then Shinobi can be responsible for maintaining only the fragments of these indexes that are important for the current workload patterns.
}

\stitle{Physical data organization on replicas. } 
Previous works also considered the idea of diverging the physical organization of replicated data. 
The technique of Fractured Mirrors~\cite{Ramamurthy2003} builds 
a mirrored database that stores its base data in 
a different physical organizations on disk 
(specifically, in a row-based and  a column-based organization). 
Similarly, Distorted Mirrors \cite{Solworth1991} presents logically 
but not physically identical mirror disks for replicated data. 
However, they do not consider how to tune the indexes for each mirror.

There are recent works~\cite{Jindal2011, Dittrich2012}
that explore different physical designs for different replicas in
Hadoop context. 
Specifically, 
TROJAN HDFS~\cite{Jindal2011} organizes
each physical replica of an HDFS block
in a different data layout, where
each data layout corresponds to a different vertical partitioning. 
Likewise, HAIL~\cite{Dittrich2012}
organizes each physical replica of an HDFS block
in a different sorted order, which essentially amounts to exactly one clustered index per replica. 
Our work differs from these works in several aspects. 
First, these works 
do not consider the problem of spreading the load evenly among specialized replicas
that we are considering. 
Second, 
performance is very sensitive to failures, because the tuning options considered by 
these papers lead to replicas that are highly specialized for subsets of the workload. 
When a replica fails, the corresponding queries will be rerouted to replicas with 
little provision to handle the rerouted workload, 
and hence performance may suffer. 
In contrast, we focus on divergent designs 
that directly take into account the possibility of replicas failing, 
thus offering more stable performance when one or more replicas become unavailable.
Third, \cite{Dittrich2012}
creates one index per replica which restricts the extent 
to which we can tune each replica to the workload. 
Our methods do not have any such built-in limitations and are 
only restricted by configurable constraints on the materialized indexes 
(e.g., total space consumption, or total maintenance cost).
Fourth, our work can return 
a set of possible designs that represent
trade-off points within a multi-dimensional space,
e.g., between workload evaluation cost and design-materialization cost. 
These works  do not support this functionality.

The original study on divergent-design tuning~\cite{Divergent2012}
introduced the $\DIVGDESIGN$ advisor 
which is the direct competitor to our proposed \rita\ advisor. 
However, as we discussed in Section~\ref{sec:intro}, 
$\divgdesign$ is fundamentally limited by the functionality of the 
underlying single-system advisor, 
and cannot support 
many essential
tuning functionalities
as \rita .

%% file: problem.tex
\section{Divergent Design Tuning: Problem Statement}
\label{sec:problem}

In this section, we formalize the problem of divergent design tuning. The problem statement borrows several concepts from the original problem statement in~\cite{Divergent2012} but also provides a non-trivial generalization. A comparison to the original study appears at the end of the section.

\subsection{Basic Definitions}

We consider a database comprising tables $T_1, \dots, T_n$. 
An \emph{index configuration} $X$ is a set of indexes defined over the database tables. 
We assume that $X$ is a subset of a universe of candidate indexes 
$\Iset = \Iset_1 \ \cup \ \cdots \ \cup \ \Iset_n$, where $\Iset_i$ represents 
the set of candidate indexes on table $T_i$. 
Each $\Iset_i$ represents a 
very large set of 
indexes and 
can be derived manually by the DBA or 
by mining the query logs. We do not place any limitations on the indexes regarding their type
or the type or count of attributes that they cover,
except that each index in $X$ is defined on exactly one table (i.e., no join indexes).\tightpara

We use $\cost(q,X)$ to denote the cost of evaluating
query~$q$ assuming that $X$ is materialized. 
The cost function can be evaluated efficiently 
in modern systems (i.e., without materializing $X$) using a 
\emph{what-if optimizer}~\cite{ChaudhuriN98}. 
We define $\cost(u,X)$ similarly for an update statement $u$, 
except that in this case we also consider the overhead of maintaining the indexes in $X$ due to the update. 
Following common practice~\cite{Schnaitter2009,Dash2011}, 
we break the execution of $u$ into two orthogonal components: 
(1) a query shell $q_{sel}$ that selects the tuples to be updated,
and (2) an update shell that performs the 
actual update on base tables and also updates any affected materialized indexes. Hence, the total cost of 
an update statement can be expressed as 
$cost(u, X) = cost(q_{sel}, X) + \sum_{a \in X}ucost(u,a)  + c_u$, 
where $ucost(u,a)$ is the cost to update index $a$ with the effects of 
the update and can be estimated again using the what-if optimizer. 
The constant $c_u$ is simply the cost to update the base data which does not depend on $X$.

We consider a database that is fully replicated in $N$ nodes, i.e., each node $i \in [1,N]$ holds a full copy of the database.
The replicas are kept synchronized by forwarding each database update to all replicas (lazily or eagerly). 
At the same time, a query can be evaluated by any replica. Since we are dealing with a multi-node system, 
we have to take into account the possibility of replicas failing. 
We use $\alpha$ to denote the probability of at least one replica failing.  
\edit{
Setting this parameter can be done once in the beginning to the best of the DBA's ability 
and then it can be updated with easy statistics as the system is used 
(you adjust it based on the failure rate you see).}
{and assume that it can be set by the DBA based on expectations about system reliability.} 
To simplify further notation, we will assume that at most one replica can fail at any point in time. The extension to multiple replicas failing together is straightforward for our problem.\tightpara

We define $\workload=Q \cup U$ as a
workload comprising a set $Q$ of query statements and a set $U$ of update statements.
Workload $\workload$ serves as the representative workload for
tuning the system. As is typical in these cases, we also define a weight function $f: \workload \rightarrow \Re$ such that $f(x)$ corresponds to the importance of query or update statement $x$ in $\workload$. The input workload and associated weights can be hand-crafted by the DBA or they can be obtained automatically, e.g., by analyzing the query logs of the database system. 

\subsection{Problem Statement}

At a high level, a divergent design allows each replica 
to have a different index configuration, tailored to 
a particular subset of the workload.
To evaluate the query workload $Q$,
an ideal strategy would route each $q \in Q$ 
to the replica that minimizes 
the execution cost for $q$. 
However, this ideal routing may not be feasible for several reasons, 
e.g., the replica may not be reachable or may be overloaded. 
Hence, the idea is to have {\it several}
low-cost replicas for $q$, so as to provide some flexibility for
query evaluation.
For this purpose, we introduce a parameter $m \in [1, N]$, which
we term {\it routing multiplicity factor}. 
Informally, for every query $q \in Q$,
a divergent design specifies a set of $m$
low-cost replicas that $q$ can be routed to.
The value of $m$
is assumed to be set by the administrator
who is responsible for tuning the system:
$m=1$ leads to a design that favors specialization;
$m=N$ provides for maximum flexibility;
$1 < m < N$ achieves some trade-off between the two extremes.

Formally, we define a divergent design as 
a pair $(\configuration, \routing)$. 
The first component $\configuration = (\confcomponent{1}, \dots, \confcomponent{N})$ 
is an $N$-tuple, where $I_r$ is the index configuration of replica $r \in [1,N]$.
The second component $\routing = (\rcomponent{0}, \rcomponent{1}, \cdots, 
\rcomponent{N})$
is a $(N+1)$-tuple of {\it routing functions}. Specifically, 
$\route{0}()$ is a function
over queries such that
$\route{0}(q)$ specifies the set of $m$ replicas to which 
$q$ can be routed when all replicas are operational (i.e., there are no failures). 
Intuitively, $\route{0}(q)$ 
indicates the replicas that can evaluate $q$ 
at low cost while respecting other constraints 
(e.g., bounding load skew among replicas, which we discus later), 
and is meant to serve as a hint to the runtime query scheduler.
Therefore, a key requirement is that $\route{0}()$ 
can be evaluated on any query $q$ and not just the queries in the training workload.
The remaining functions $\route{1}, \dots, \route{N}$ have a similar functionality 
but cover the case when replicas fail: 
$\route{j}()$, for $j \in [1,N]$, 
specifies how to route each query when replica $j$ has failed and is not reachable. Notice that in this case there may be fewer than $m$ replicas in $\route{j}(q)$ for any $q \in Q$ if the DBA has originally specified $m=N$.

In order to quantify the goodness of a divergent design,
we first use a metric that
captures the performance of the workload
under the normal operation when no running replica fails as follows.
\begin{eqnarray*}
	\totalcost(\configuration, \routing) & = & \sum_{q \in Q} \sum_{r \in \rcomponent{0}(q)} 
				\frac{f(q)}{m} \cost(q, \confcomponent{r}) + \\
	& & \sum_{ u \in U} \sum_{i \in [1,N]} f(u) \cost(u, \confcomponent{i})
\end{eqnarray*}
\noindent
The second term simply captures the cost 
to propagate each update $u \in U$ to each replica in the system. 
The first summation captures the cost to evaluate the query workload $Q$.
We assume that $q$ is routed uniformly among its $m$ replicas in $\rcomponent{0}(q)$, 
and hence the weight of $q$ is scaled by $1/m$ for each replica. 
The intuition behind the $\totalcost(\configuration, \routing)$ 
metric is that it captures the ability of the divergent design 
to achieve both replica specialization and flexibility in load balancing
with respect to $m$. 

To capture the case of failures, we define $\ftotalcost(\configuration, \routing,j)$ as the performance of the workload when replica~$j \in [1,N]$ fails:
\begin{eqnarray*}
        \ftotalcost(\configuration, \routing,j) & = & \sum_{q \in Q} \sum_{r \in \route{j}(q)} 
                                \frac{f(q)}{\max\{m,N-1\}} \cost(q, \confcomponent{r}) + \\
        & & \sum_{ u \in U} \sum_{i \in \{ 1, \cdots, N \} - \{ j \} } f(u) \cost(u, \confcomponent{i})
\end{eqnarray*}
The expression for $\ftotalcost(\configuration, \routing, j)$
is similar to $\totalcost(\configuration, \routing)$, except that, since replica $j$ is unavailable,  the update cost on replica $j$ is discarded and routing function $\route{j}$ is used
instead of $\route{0}$.

We quantify the goodness of a divergent design $(\configuration, \mapping)$
based on the {\it expected cost} of the workload, denoted as $\expected(\configuration,\mapping)$, by combining $\totalcost(\configuration,\mapping)$ and $\ftotalcost(\configuration,\mapping,j)$ weighted appropriately. Recall that $\alpha$ is a DBA-specified probability that a failure will occur. It follows that $(1-\alpha)$ is the probability that all replicas are operational and hence the performance of the workload is computed by $\totalcost(\configuration,\mapping)$. Conversely, the probability of a specific replica $j$ failing is $\alpha/N$, assuming that all replicas can fail independently with the same probability. In that case, the cost of workload evaluation is $\ftotalcost(\configuration,\mapping,j)$. Putting everything together, we obtain the following definition for the expected workload cost:
\begin{eqnarray*}
	\expected(\configuration,\mapping) & = & (1 - \alpha) \cdot \totalcost( \configuration, \routing) + \\
		 & &	   \sum_{j \in [1, N]} \frac{\alpha}{N} \ftotalcost(\configuration, \routing,j)
\end{eqnarray*}
\noindent
Our assumption so far is that at most one replica can be inoperational at any point in time. The extension to concurrent failures is straightforward. All that is needed is extending $\mapping$ with routing functions for combinations of failed replicas, and then extending the expression of $\expected(\configuration,\mapping)$ with the corresponding cost terms and associated probabilities.

We are now ready to formally define the problem of {\bf D}ivergent {\bf D}esign {\bf T}uning, 
referred to as \divergent. 

\begin{problem}(Divergent Design Tuning - \divergent) \label{prob:div}
We are given a replicated database with $N$ replicas, a workload $\workload = Q \ \cup \ U$, 
a candidate index-set $\Iset$, a set of constraints $C$, a routing multiplicity factor $m$, 
and a probability of failure $\alpha$. 
The goal is to compute a divergent design $(\configuration, \mapping)$ that employs indexes in $\Iset$, satisfies the constraints in $C$, and $\expected(\configuration,\mapping)$ is minimal among all feasible divergent designs. 
\eop
\end{problem}

\stitle{Constraints in \divergent.} The set of constraints $C$ enables the DBA to control the space of divergent designs considered by the advisor. An \emph{intra-replica} constraint specifies some desired property that is local to a replica. Examples include the following:
\begin{itemize}
	\item The size of $I_j$ in $\configuration$ is within a storage-space budget.
	\item Indexes in $I_j$ must have specific properties, e.g., no index can be more than 5-columns wide, or the count of multi-key indexes is below a limit.
	\item The cost to update the indexes in $I_j$ is below a threshold.
\end{itemize}
Conversely, an \emph{inter-replica} constraint specifies some property that involves all the replicas. Examples include the following:
\begin{itemize}
	\item If $(\configuration_c,\mapping_c)$ represents the current divergent design of the system, then $\expected(\configuration,\mapping)$ must improve on $\expected(\configuration_c,\mapping_c)$ by at least some percentage.
	\item The total cost to materialize $(\configuration,\mapping)$ (i.e., to build each $I_j$ in each replica) must be below some threshold.
	\item The load skew among replicas must be below some threshold. (We discuss this constraint in more detail shortly.)
\end{itemize}
We will formalize later the precise class of constraints $C$ that we can support in $\rita$. The goal is to provide support for a large class of practical constraints, while retaining the ability to find effective designs efficiently. 

Bounding load skew is a particularly important inter-replica constraint that we examine in our work. The replica-specialization imposed by a divergent design means that each replica may receive a different subset of the workload, and hence a different load. The $\expected()$ metric does not take into account these different loads, which means that minimizing workload cost may actually lead to a high skew in terms of load distribution. Our experiments verify this conjecture, showing that an optimal divergent design in terms of $\expected()$ can cause loads at different replicas to differ by up to a factor of two. This situation, which is clearly detrimental for good performance in a distributed setting, can be avoided by including in $C$ a constraint on the load skew among replicas. 
More concretely, the load of replica $j$ under normal operation can be computed as: 
\begin{equation*}
	\load(\configuration, \mapping, j) =  \sum_{q \in Q \land\ j \in \route{0}(q)} 
				\frac{f(q)}{m} \cost(q, \confcomponent{j}) + 
				 \sum_{ u \in U} f(u) \cost(u, \confcomponent{j})
\end{equation*}
We say that design $(\configuration, \mapping)$ has \emph{load skew} 
${\nodefactor \ge 0 }$ 
if and only if $\load( \configuration,\mapping, r) \le (1+\nodefactor) \cdot \load(\configuration,\mapping, j)$
for any $1 \le r \ne j \le N$. 
A low value is desirable for $\nodefactor$, as
it implies that $(\configuration, \mapping)$ keeps the different replicas relatively balanced. 

We can define a load-skew constraint for the case of failures in exactly the same way. 
Specifically,
we define $\loadfail(\configuration, \routing, j, f)$ as the load of replica $j$
when replica $f$ fails.
The formula of $\loadfail(\configuration, \routing, j, f)$ 
is similar to that of $\load(\configuration, \routing, j)$
except that $\rcomponent{0}$ is replaced by $\rcomponent{f}$. The constraint then specifies that $\loadfail(\configuration,j,f) \le (1+\nodefactor')\loadfail(\configuration,\routing,r,f)$ for any valid choice of $j,r,f$ and a skew factor $\nodefactor' \ge 0$. 

It is straightforward to verify that zero skew is always possible by assigning the same index configuration to each replica. One may ask whether there is a tradeoff between specialization (and hence overall performance) and a low skew factor. One of the contributions of our work is to show that this is \emph{not} the case, i.e., it is possible to compute divergent designs that exhibit both good performance and a low skew factor. 

\stitle{Theoretical Analysis.} Computing the optimal divergent design implies computing a partitioning of the workload to replicas and an optimal index configuration per replica. Not surprisingly, the problem is computationally hard, as formalized in the following theorem.
The proof is provided in Appendix~\ref{app:hard}.  

\begin{theorem}\label{theorem:hard}
It is not possible to compute an optimal solution to $\divergent$ in polynomial time unless $P = NP$. 
\end{theorem}

\subsection{Comparison to Original Study~\cite{Divergent2012}}
\label{sec:compare}

The formulation of \divergent\ expands on the original problem statement in~\cite{Divergent2012} in several non-trivial ways. First, \divergent\ incorporates the expected cost under the case of failures into the objective function, whereas failures were completely ignored in~\cite{Divergent2012}. Second, our formulation allows a much richer set of constraints $C$ compared to the original study which considered solely intra-replica constraints. As discussed earlier, the omission of such constraints may lead to divergent designs with undesirable effects on the overall system, e.g., the load skew issue that we discussed earlier. Finally, the original problem statement imposed a restriction for $\rcomponent{0}(q)$ to correspond to the $m$ replicas with the least evaluation cost for $q$, that is, $\forall q \in Q$ and $\forall i, j \in [1,N]$ such that $i \in \rcomponent{0}(q)$ and $j \notin \rcomponent{0}(q)$ it must be that  $cost(q, \confcomponent{i}) \leq cost(q, \confcomponent{j})$. We remove this restriction in our formulation in order to explore a larger space of divergent designs, which is particularly important in light of the richer class of constraints that we consider.

%% file: divergent.tex
\section{Divergent Design Tuning as Binary Integer Programming}
\label{sec:div-bip}

In this section, we show that the problem of Divergent Design Tuning ($\divergent$)
can be reduced to a {\it Binary Integer Program (BIP)} that contains a relatively small number of variables. The implication is that we can leverage several decades of research in linear-optimization solvers in order to efficiently compute near-optimal divergent designs. Reliance on these off-the-shelf solvers brings other important benefits as well, e.g., simpler implementation and higher portability of the index advisor, or the ability to operate in ``any-time'' mode where the DBA can interrupt the tuning session at any point in time and obtain the best design computed thus far. We discuss these features in more detail in Section~\ref{sec:tool}, when we describe the architecture of $\rita$. 

The remainder of the section presents the technical details of the reduction. We first review some basic concepts for \emph{fast what-if optimization}, which forms the basis for the development of our results. We then present the reduction for a simple variant of $\divergent$ and then generalize to the full problem statement.

\subsection{Fast What-If Optimization} 
\label{sec:inum}

What-if optimization is a principled method to estimate $\cost(q,X)$ and $\cost(u,X)$ for any $q\in Q$, $u \in U$ and index set $X$, but it remains an expensive operation that can easily 
become the bottleneck in any index-tuning tool. 
To mitigate the high overhead of what-if optimization, 
recent studies have developed two techniques for 
fast what-if optimization, termed INUM~\cite{Papadomanolakis2007} and C-PQO~\cite{Bruno2008} respectively, that can be used as 
drop-in replacements for a what-if optimizer. 
In what follows, we focus on INUM but note
that the same principles apply for C-PQO.

We first introduce some necessary notation. A configuration $A \subseteq \Iset$ is called {\it atomic}~\cite{Papadomanolakis2007} 
if $A$ contains at most one index from each $\Iset_i$. 
We represent $A$ as a vector with $n$ elements, 
where $A[i]$ is an index from $\Iset_i$ or 
a symbol $\noindex{i}$ indicating that no index of $\Iset_i$ is selected. 
For an arbitrary index set $X$,  
we use $atom(X)$ to denote the set of atomic configurations 
in $X$. To simplify presentation, 
we assume that a query $q$ references a specific table $T_i$ 
with at most one tuple variable. 
The extension to the general case is straightforward at the expense of complicated notation.

For each query $q$, INUM makes a few carefully selected calls to the what-if optimizer 
in order to compute a set of \emph{template plans}, denoted as $\tplans(q)$. A template plan $p \in \tplans(q)$ is a physical plan for $q$ except
that all access methods (i.e., the leaf nodes of the plan)
are substituted by ``slots''. Given a template $p \in \tplans(q)$ and an atomic index configuration $A$, 
we can instantiate a concrete physical execution plan by instantiating each slot with the corresponding index in $A$, 
or a sequential scan if $A$ does not prescribe an index for the corresponding relation. Figure~\ref{fig:inum} shows an example of this process for a simple query \add{over three tables $T_1$, $T_2$, and $T_3$}, 
and an atomic configuration that specifies \add{an index on $T_1$ and another index on $T_3$}. 
Each template is also associated with an {\it internal plan cost},  
which is the sum of the costs of the operators in this plan except the access methods. 
Given an atomic configuration $A$, the cost of the instantiated plan, denoted as $\cost(p,A)$, is the sum of the internal plan cost and the cost of the instantiated access methods.

The intuition is that $\tplans(q)$ represents the possibilities for the optimal plan
of $q$ depending on the set of materialized indexes. Hence, given a hypothetical index
configuration $X$, INUM estimates $\cost(q,X)$ as the minimum $\cost(p,A)$ over $p \in
\tplans(q)$ and $A \in \atom(X)$. Note that a slot in $p$ may have restrictions on its
sorted order, e.g., the template plan in Figure~\ref{fig:inum} prescribes that the
slot for $T_1$ must be accessed in sorted order of attribute $x$. If $A$ does not
provide a suitable access method that respects this sorted order, then $\cost(p,A)$ is
set to $\infty$. INUM guarantees that there is at least one plan $p$ in $\tplans(q)$
such that $\cost(p,A) < \infty$ for any $A \in \atom(X)$. As shown in the original
study~\cite{Papadomanolakis2007}, INUM provides an accurate approximation for the
purpose of index tuning, and is orders-of-magnitude faster compared to conventional
what-if optimization.

\begin{figure}
	\centering
	\includegraphics[width=0.45\textwidth]{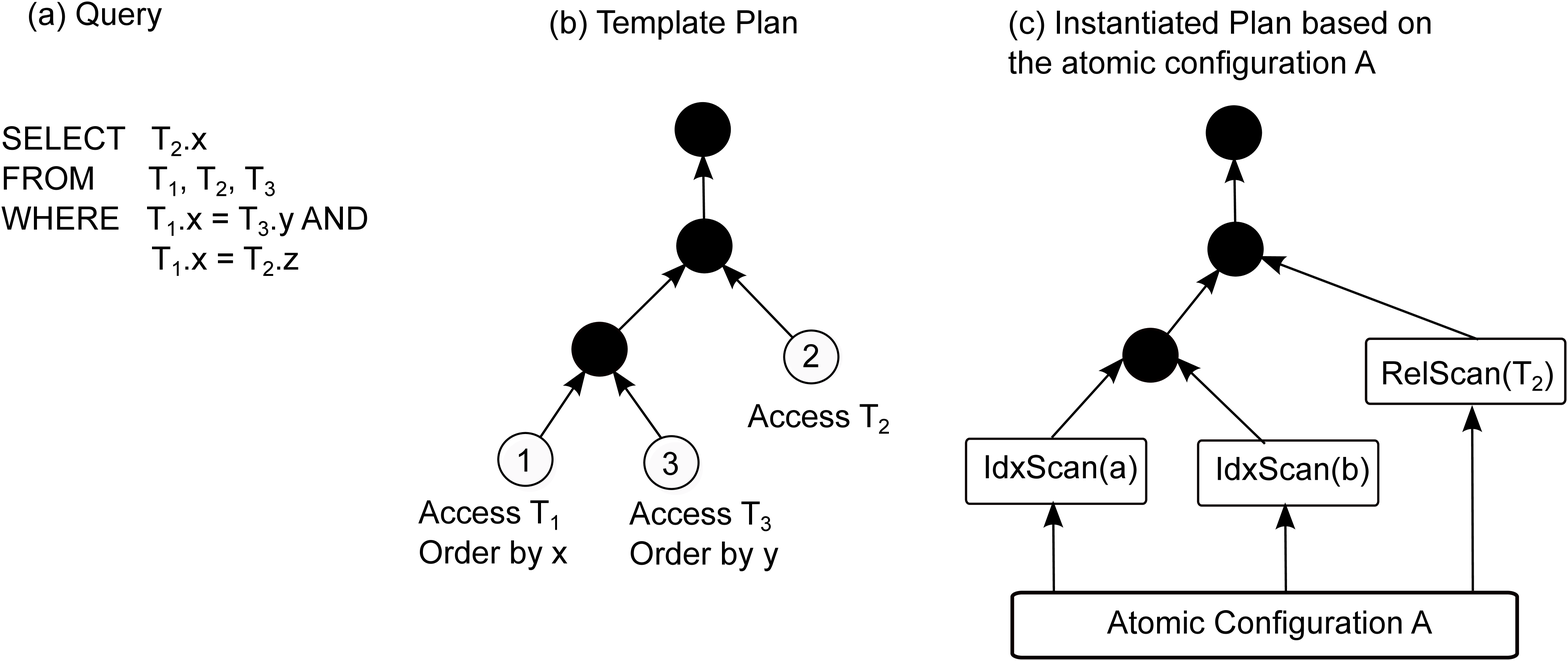}
	\caption{\add{Example of template plans and instantiated plans.
		The configuration $A$
		has the following contents:
		$A[1]= a$, an index with key $T_1.x$;
		$A[2] = \noindex{2}$;
		$A[3] = b$, an index with key $(T_2.x, T_2.w)$~\cite{Dash2011}}}
	\label{fig:inum}
\end{figure}    

\stitle{Linear composability.}
The approximation provided by INUM and C-PQO can be formalized in terms of a property 
that is termed {\it linear composability} in~\cite{Dash2011}.

\begin{definition}[Linear composability~\cite{Dash2011}] \label{def:linear}
Function $\cost()$ is linearly composable for a select-statement $q$
if there exists a set of identifiers $K_q$ and constants
$\beta_{p}$ and $\gamma_{pa}$ for $p \in  K_q$,
$a \in \Iset \cup \{ \noindex{1} \} \cup \cdots \cup \{ \noindex{n} \}$ such that:
\[cost(q, X) = min\{ \beta_{p} + 
                \sum_{a \in A} \gamma_{pa}, p \in K_q, A \in \atom(X)\}
       \] for any configuration X.
Function $cost()$ is linearly composable for an update-statement $q$
if it is linearly composable for its query shell.  \eop
\end{definition}

It has been shown in~\cite{Dash2011}
that both INUM and C-PQO compute a cost function that is linearly composable. 
For INUM, $K_q = \tplans(q)$
and each $p$ corresponds to a distinct template plan in $\tplans(q)$.
Here, we use $\tplans(q)$
for the set of identifiers and overload $p \in \tplans(q)$
to represent an identifier.
In turn, the expression $\beta_{p} +  \sum_{} \gamma_{pa}$
corresponds to $cost(p,A)$,
where $\beta_{p}$ denotes the internal plan cost of $p$,
and $\gamma_{pa}$ 
is the cost of implementing the corresponding slot in $p$ using index $a$.
(The slot covers the relation on which the index is defined.)
Note that linear composability does not imply a linear cost model for the query optimizer
-- non-linearities are simply hidden inside the constants $\beta_{qp}$.

For the remainder of the paper, we assume that $\cost(q,X)$ is computed by either INUM or C-PQO (for the purpose of fast what-if optimization) and hence respects linear composability.

\subsection{Basic DDT}
\label{sec:basic}

In this subsection, we discuss how to reduce $\divergent$ to a compact BIP for the case when $\alpha = 0$, $C = \emptyset$ (i.e., no failures and no constraints) and the workload comprises solely queries, i.e., $W = Q$. This reduction forms the basis for generalizing to the full problem statement, which we discuss later.

\begin{figure}[t]
{\small
  Minimize:   
  $\biptotalcost(\configuration, \routing)
  				= \bipquerycost (\configuration, \routing) 
  				\boxed{+
  				\bipupdatecost (\configuration, \routing)}$, 
	\\ where:
	  \begin{equation*} \label{eq:query-cost}	  
		\bipquerycost (\configuration, \routing) 
				 = \sum_{ q \in Q} \sum_{r \in [1,N]} 
					 \frac{f(q)}{m} \costbip(q, r)			
		\end{equation*}					 
	\begin{empheq}[box=\fbox]{align*}
		\bipupdatecost (\configuration, \routing) 
			& = & \sum_{ q \in Q_{upd}} \sum_{r \in [1,N]} 							f(q) \costbip(q, r)	  \\
				 & + &
		\sum_{ u \in U} \sum_{r \in [1,N]} f(u) s^{r}_a \cdot ucost(u,a)
		\end{empheq}
 \begin{equation} \label{eq:cost-formula}
	\costbip(q, r) = 
		\sum_{p \in \TPlans{q}}\beta_{p}y^{\rep}_{p} + 
		\sum_{\substack{p \in \TPlans{q} \\ a \in \Iset \cup 
			\{ \noindextiny{1} \} \cup \cdots \cup \{ \noindextiny{n} \}
			}}
		\gamma_{pa}x^{\rep}_{pa}, 		
		\substack{\forall r \in [1, N],  \\ \forall q \in Q \cup Q_{upd}}
 \end{equation}
 such that:	
	\begin{equation}
		\sum_{\rep \in [1, N]} t^r_q = m, \forall q \in Q
		\label{eq:div-top-m-q}		
	\end{equation}
	\begin{empheq}[box=\fbox]{equation}
		\sum_{\rep \in [1, N]} t^r_q = N, \forall q \in Q_{upd}
		\label{eq:div-top-m-upd}		
	\end{empheq}
	\begin{equation}
		\sum_{p \in \TPlans{q}} y^{r}_{p} = t^r_q,\ \ \ \forall q \in Q \cup Q_{upd}
		 \label{eq:div-one-template} 		
	\end{equation}
  \begin{equation} \label{eq:recommend-index}
  	s^{\rep}_a  \geq x^{\rep}_{pa},
		\ \ \ \forall q \in Q \cup Q_{upd}, 
			 p \in \TPlans{q}, \  a \in \Iset
  \end{equation}
 \begin{equation} \label{eq:div-atomic}
     \sum_{a \in \Iset_i \cup \{ \noindextiny{i} \} }x^{\rep}_{pa} = y^{\rep}_{p}, 
		\ \ \substack{ \forall q \in Q \cup Q_{upd},  p \in \TPlans{q}, \\
						  i \in [1,n], \ T_i \mbox{ is referenced in q} 
					 }	 
 \end{equation} 
 
}
\caption{The BIP for Divergent Design Tuning. \label{fig:bip_basic}}
\end{figure}

\stitle{BIP formulation.} At a high level, we are given an instance of $\divergent$ and we wish to construct a BIP whose solution provides an optimal divergent design. This reduction will hinge upon the linear composability property, i.e., we assume that each query $q \in W$ 
has been preprocessed with INUM and therefore we can approximate 
cost$(q,X)$ for any $X \subseteq \Iset$ 
as expressed in Definition~\ref{def:linear}. 

Figure~\ref{fig:bip_basic} shows the constructed BIP. (Ignore for now the boxed expressions.) In what follows, we will explain the different components of the BIP and also formally state its correctness. The BIP uses two sets of binary variables to encode the choice for a divergent design $(\configuration,\routing)$: 
\begin{itemize}
	\item Variable $s^r_a$ is set to $1$ if and only if index $a$
	is part of the index design $I_r$ on replica $r$. In other words, $I_r = \{ a \ | \ s^r_a = 1 \}$.
	
	\item Variable $t^r_q$ is set to $1$ if and only if query $q$ is routed to replica $r$, i.e., $r \in \route{0}(q)$. (Recall that we ignore failures for now.) In other words, $\route{0}(q) = \{ r \ | \ t^r_q =1 \}$. 
\end{itemize}
Under our assumption of using fast what-if optimization, the cost of a query $q$ in some replica $r$ can be expressed as $\cost(q,I_r) = \cost(p',A')$ for some choice of $p' \in \tplans(q)$ and an atomic configuration $A' \in \atom(I_r)$. To encode these two choices, we introduce two different sets of binary variables:
\begin{itemize}
	\item Variable $x^r_{pa}$, where $p$ is a template in $\tplans(q)$ and $a$ is an index in 
$\Iset$$\cup \{ \noindex{1} \} \cup \cdots \cup \{ \noindex{n} \}$, 
is equal to $1$ if and only if $p=p'$ and $a \in A'$.
	\item Variable $y^r_p = 1$ if and only if $p=p'$. 
\end{itemize} 
The BIP specifies several constraints that govern the valid value assignments to the aforementioned variables:
\begin{itemize}
	\item Constraint~(\ref{eq:div-top-m-q}) specifies that query $q$ must be routed to exactly $m$ replicas. 
	\item Constraint~(\ref{eq:div-one-template}) specifies that there must be exactly one variable $y^r_p$ set to $1$ if $t^r_q=1$, i.e., exactly one template $p$ chosen for computing $\cost(q,I_r)$ if $q$ is routed to $r$. Conversely, $y^r_p=0$ for all templates $p$ if $t^r_q = 0$. 
	\item Constraint~(\ref{eq:recommend-index}) specifies that an index $a$ can be used in instantiating a template $p$ at replica $r$ only if it appears in the corresponding design $I_r$.
	\item Constraint~(\ref{eq:div-atomic}) specifies that if $y^r_p=1$, i.e., $p$ is used to compute $\cost(q,I_r)$, then there must be exactly one access method $a$ per slot such that $x^r_{pa}=1$. Essentially, the choices of $a$ for which $x^r_{pa}=1$ must correspond to an atomic configuration. Conversely, $x^r_{pa}=0$ for all $a$ if $y^r_p=0$.
\end{itemize}

Given these variables, we can express $\cost(q,I_r)$ as in Equation~\ref{eq:cost-formula} in Figure~\ref{fig:bip_basic}. The equation is a restatement of linear composability (Definition~\ref{def:linear}) by translating the minimization to a guarded summation using the 
binary variables $y^r_p$ and $x^r_{pa}$. 
Specifically, if $t^r_q=1$, then constraint~\eqref{eq:div-one-template} 
forces the solver to pick exactly one $p$ such that $y^r_p=1$, and constraint~\eqref{eq:div-atomic} forces setting 
$x^r_{pa}=1$ for the same choice of $p$ and corresponding to an atomic configuration. Hence, minimizing the expression in Equation~\ref{eq:cost-formula} corresponds to computing $\cost(q,I_r)$. Otherwise, if $t^r_q=0$, then the same constraints force $\cost(q,I_r)=0$. In turn, it follows that the objective function of the BIP corresponds to $\TotalCost(\configuration,\routing)$.

\stitle{Handling update statements.}
The total cost to execute update
statements, $\updatecost(\configuration, \routing)$, 
includes two terms, 
as shown in the second boxed expression in Figure~\ref{fig:bip_basic}.
Here, $Q_{upd}$ denotes the set of all the query-shells,
each of which corresponds to each update statement in $U$. 
The first component of $\updatecost()$
is the total cost
to evaluate every query-shell in $Q_{upd}$
at every replica.
This component is expressed as the summation of $\costbip(q, r)$
for all $q_{sel} \in Q_{upd}$ and $r \in [1,N]$
in our BIP. 
Since each query-shell needs to be routed to all replicas,
we impose the constraint~\eqref{eq:div-top-m-upd}.

The second component of $\updatecost()$ is
the total cost to update the affected indexes.
Using variable $s_a^r$ that tracks the selection
of an index at replica $r$ in the recommended
configuration, the cost of updating an index $a$ 
at replica $r$
given the presence
of an update statement $u$ is computed as the product of $s^r_a$
and $ucost(u,a)$.

\stitle{Correctness.} Up to this point, we argued informally about the correctness of the BIP. The following theorem formally states this property. 
The proof is given in Appendix~\ref{app:div-equivalent}.

\begin{theorem}\label{theorem:div-correctness}
A solution to the BIP in Figure~\ref{fig:bip_basic} corresponds to the optimal divergent design for $\divergent$ when $\alpha=0$ and $C=\emptyset$. 
\end{theorem}

\eat{
\begin{proof} (Sketch)
We prove the theorem in two steps. 
First, we show that every divergent design $(\configuration, \routing)$
corresponds to a value-assignment $\va$
for variables in the BIP 
such that $\va$ satisfies the constraints. This property guarantees that
the solution space of the BIP contains all possible solutions for the
divergent design tuning problem. 
Subsequently, we prove that the optimal assignment $\va^*$ corresponds
to a divergent design.
Combining these two results, we can then conclude the
correctness of the theorem. 
\end{proof}
}

As stated repeatedly, the key property of the BIP is that it contains a relatively small number of variables and constraints, which means that a BIP-solver is likely to find a good solution efficiently. Formally:
\begin{corollary} \label{co:var}
The number of variables and constraints
in the BIP shown in Figure~\ref{fig:bip_basic} is
in the order of $O(N |\workload| |\Iset|)$.
\end{corollary}
In fact, it is possible to eliminate some variables and constraints from the BIP while maintaining its correctness. We do not show this extension since it does not change the order of magnitude for the variable count but it makes the BIP less readable and harder to explain.


\input{robustness}

\subsection{Routing Queries}
\label{sec:routing}

Recall that a divergent design $(\configuration,\routing)$ includes both the index-sets for different replicas and the routing functions $\route{0}(), \route{1}(),\dots,\route{N}()$. These functions are used at runtime, after the divergent design has been materialized, to route queries to different specialized replicas. A solution to the BIP determines how to compute these functions for a training query $q$ in $Q$, based on the variables $t^r_q$ and $t^{r,j}_q$. Here, we describe how to compute these functions for any query $q'$ that is not part of the training workload. We focus on the computation of $\route{0}(q')$ but our techniques readily extend to the other functions. 

Our first approach is inspired by the original problem statement of the tuning problem~\cite{Divergent2012} and computes $\route{0}(q')$  as the $m$ replicas with the lowest evaluation cost for $q'$. Normally this requires $N$ what-if optimizations for $q'$, but we can leverage again fast what-if optimization in order to achieve the same result more efficiently. Specifically, we first compute $\tplans(q')$ (which requires a few calls to the what-if optimizer) and then formulate a BIP that computes the top $m$ replicas for $q'$. 

Our second approach tries to match more closely the revised problem statement, where a query is not necessarily routed to its top $m$ replicas. Our approach is to match $q'$ to its most ``similar'' query $q$ in the training workload $Q$, and then to set $\route{0}(q')=\route{0}(q)$. The intuition is that the two queries would affect the divergent design similarly if they were both included in the training workload. We can use several ways to assess similarity, but we found that fast what-if optimizations provides again a nice solution. Specifically, we compute again $\tplans(q')$ and then quickly find the optimal plan for $q'$ in each replica. We then form a vector $v_{q'}$ where the $i$-th element is the set of indexes in the optimal plan of $q'$ at replica $i$. We can compute a similar vector for $v_{q}$ and then compute the similarity between $q'$ and $q$ as the similarity between the corresponding vectors\footnote{Any vector-similarity metric will do. We first convert $v_{q'}$ $v_{q}$ to binary vectors indicating which indexes are used at each replica and then use a cosine-similarity metric.}. The intuition is that $q'$ is similar to $q$ if in each replica they use similar sets of indexes. We can refine this approach further by taking into account the top-$2$ plans for each query, but our empirical results suggest that the simple approach works quite well.

%% file: robustness.tex
\subsection{Factoring Failures}
\label{sec:failure}

To extend the BIP to the case when $\alpha>0$ (i.e., failures are possible), we first introduce additional variables 
$t^{r,j}_q$, $y^{r, j}_{p}$ and $x^{r, j}_{pa}$, for $j \in [1,N]$. 
These variables have the same meaning as their counterparts in Figure~\ref{fig:bip_basic}, 
except that they refer to the case where replica $j$ fails. 
For instance, $t^{r,j}_q=1$ if and only if $q$ is routed to replica $r$ when $j$ fails, i.e., 
$\route{j}(q)=\{r \ | \ t^{r,j}_q=1\}$. 
We augment the BIP with the corresponding constraints as well. 
For instance, we add the constraint $\sum_{r \ne j} t^{r,j}_q = \max\{N-1, m\}$, $\forall q \in Q, j \in [1,N]$ to express the fact that function $\route{j}()$ must respect the routing-multiplicity factor $m$. Finally, we change the objective function to $\expected()$, which is already linear, and express each term $\ftotalcost(\configuration,\routing,j)$ as a summation that involves the new variables. 

The complete details for this extension, including the proof of correctness, can be found in Appendix~\ref{app:failure}. 
We should mention that this extension increases the number of variables and constraints by a factor of $N$ to $O(N^2 |\workload||\Iset|)$, since it becomes necessary to reason about the failure of every replica $j \in [1,N]$. 

\subsection{Adding Constraints}
\label{sec:constraint}

In this subsection, we discuss how to extend the BIP when $C\ne \emptyset$, i.e., the DBA specifies constraints for the divergent design. 

Obviously, we can attach to the BIP any type of linear constraint. As it turns out, linear constraints can capture a surprisingly large class of practical constraints. In what follows, we present three examples of how to translate common constraints to linear expressions that be directly added to the BIP. 

\stitle{Space budget.} Let $size(a)$ denote the estimated size
of an index $a$,
and $b$ be the storage budget at each replica.
Using the variable~$s_a^r$ that tracks the selection
of an index at replica $r$ in the recommended
configuration, the storage constraint can be encoded
as: $\sum_{a \in \Iset} s^r_a size(a) \leq b, \ \forall r \in [1, N]$. 
In general, variables $s^r_a$ can be used to express several types of intra-replica constraints that involve the selected indexes, e.g., bound the total number of multi-key indexes per replica, or bound the total update cost for the indexes in each replica.

\stitle{Bounding load-skew.} Recall that $\load(\configuration, \routing, j)$ captures the total load of replica $j$ under a divergent design $(\configuration,\routing)$. The load-skew constraint specifies that $\load(\configuration,\routing,j) \le (1+\tau) \load(\configuration,\routing,r)$, for any $r \ne j$, where $\tau$ is the load-skew factor provided by the DBA.

It is straightforward to translate the constraint between two specific replicas $j$ and $r$ 
into a linear inequality, by using variables $x^r_{pa}$ and $y^r_{p}$ to rewrite the corresponding $\load()$ terms as linear sums. 
Specifically, $\load(\configuration,\routing,j)$ can be expressed as a linear sum similarly to $\biptotalcost()$ 
in Figure~\ref{fig:bip_basic}, except that we only consider replica $j$ 
and the queries for which $j \in \route{0}(q)$, and the same goes for expressing $\load(\configuration,\routing,r)$. 

Based on this translation, we can add $N(N-1)$ constraints to the BIP, one for each possible choice of $j$ and $r$. We can actually do better, by observing that we can sort replicas in ascending order of their load, and then impose a single load-skew constraint between the first and last replica. By virtue of the sorted order, the constraint will be satisfied by any other pair of replicas. Specifically, we add the following two constraints to the BIP:
\begin{eqnarray}
     	\load(\configuration, \mapping, i) \leq  \load(\configuration, \mapping, i + 1), \ \forall i \in [1, N-1] \label{eq:imbalance-a}  \label{eq:order} \\
		\load(\configuration, \mapping, N) \leq  (1 + \nodefactor) \cdot \load (\configuration, \mapping, 1) \label{eq:imbalance-b} \label{eq:imbalance}
\end{eqnarray}
This approach requires only $N$ constraints and is thus far more effective. 

The final step requires adding another set of constraints on $\costbip(q,I_r)$. 
This is a subtle technical point that concerns the correctness of the reduction 
when the constraints are infeasible. 
More concretely, the solver may assign variables $y^r_{p}$ and $x^r_{pa}$ 
for some query $q$ so that constraints~\eqref{eq:order}--\eqref{eq:imbalance} 
are satisfied even though this assignment does not correspond to the optimal cost $\cost(q,I_r)$. 
To avoid this situation, we introduce another set of variables that are isomorphic to $x^r_{pa}$ and are used to force a cost-optimal selection for $y^r_p$ and $x^r_{pa}$. The details are given in Appendix~\ref{app:exact-balance}, but the upshot is that we need to add $O(N |\workload||\Iset|)$ additional constraints. 

We have also developed an approximate scheme to handle load-skew constraints in the BIP. The approximate scheme allows the BIP to be solved considerably faster, but the compromise is that the resulting divergent design may not be optimal. However, our experimental results (see Section~\ref{sec:expt}) suggest that the loss in quality is not substantial. The details of the approximate scheme can be found in Appendix~\ref{app:greedy}

\stitle{Materialization cost constraint. }
This constraint specifies that 
the total cost to materialize $(\configuration, \routing)$
must be below some threshold $\Cmat$. 
The materialization cost is computed with respect to the current design $(\configuration_c, \routing_c)$ and takes
into account the cost to scale up or down the current number of replicas, and the cost to create additional indexes or drop redundant indexes in each replica. 

We first consider the case when the number of replicas remains unchanged between $(\configuration,\routing)$ and $(\configuration_c, \routing_c)$. Let us consider a specific replica $r$ and the new design $I_r \in \configuration$. Let $I^c_r \in \configuration_c$ denote the previous design. Clearly, we need to create every index in $I_r - I^c_r$ and to delete every index in $I^c_r - I_r$. Assuming that $\ccost(a)$ and $\dcost(a)$ denote the cost to create and drop index $a$ respectively, we can express the reconfiguration cost for replica $r$ as $\sum_{a \not\in I^c_r} s^r_a \ccost(a) + \sum_{a \in I^c_r}(1-s^r_a)\dcost(a)$. If each replica can install indexes in parallel, then the materialization cost constraint
can be expressed as:
\begin{equation*}
	\sum_{a \in \Iset \wedge a \not \in \Icurrent_r} 
	 s^r_a \ccost(a) + \sum_{a \in \Iset \wedge a \in \Icurrent_r} (1-s^r_a) \dcost(a) \leq  \Cmat,  \forall r \in [1, N]   
\end{equation*}
We can also express a single constraint on the aggregate materialization cost by summing the per-replica costs.

We next consider the case when the DBA wants to 
shrink the number of replicas to be $\Ndeploy < N$. In this case, the BIP solver should try to find which replicas to maintain and how to adjust their index configurations so that the total materialization cost remains below threshold. For this purpose, we introduce $N$ new binary variables $\deploy^r$ with $r \in [1, N]$
associated with each replica $r$,
where $\deploy^r = 1$
if replica $r$ is kept in the new divergent design, and $\deploy^r = 0$ 
otherwise. The materialization cost can be computed in a similar way as discussed above, except that 
we need to add the following two additional constraints to the BIP.
	\begin{subequations}
		\begin{align}
			t^r_q \leq z^r, \forall q \in Q\cup Q_{upd}, r \in [1,N] \\
			\sum_{r \in [1,N]}z^r = \Ndeploy
		\end{align}
	\end{subequations}
The first constraint ensures that we can route queries only to live replicas.
The second simply restricts  the number of live replicas to the desired number.

Lastly, we consider the case 
when the DBA wants to expand
the number of replicas to be $\Ndeploy > N$.
The set of constraints in the BIP can be re-used 
except that all the variables 
are defined according to $\Ndeploy$ replicas (instead of $N$ replicas as before). 
The materialization cost can also be computed in a similar way.
In addition, we also take into account the cost to deploy the database in new replicas, which appear as constants in the total cost to materialize a design in a new replica.

\eat{
In the following, we propose a greedy scheme 
that trade-offs the quality of the design for the efficiency. 
In our experiments, this solution runs five times faster
than the above exact solution, while returning 
the divergent designs that have $\expected$ as good as the ones
returned by the exact solution.

First, we derive 
an optimal design $(\configuration_{opt}, \routing_{opt})$
assuming there is no load imbalance constraint
and the probability of failure is $0$. 
We then compute
an approximation factor $\beta = \frac{\nodefactor - 1}{1 + (N - 1) \nodefactor}$.
and add
the following constraint into the BIP. 
\begin{equation} \label{eq:div-greedy}
\load(\configuration, \mapping, r) \leq \frac{(1 + \beta) \totalcost(\configuration_{opt}, \mapping_{opt})}{N}, 
\forall r \in [1,N]
\end{equation}

This constraint is much easier than the constraint 
in~\eqref{eq:imbalance-exact}, 
as its right handside is a constant.
We prove that if the BIP solver can find a solution
for the modified BIP,
the returned solution is a valid solution 
and has $\totalcost(\configuration, \routing)$ bounded
as the following theorem shows.

\begin{theorem} \label{theorem:node-factor}
The divergent design returned by the greedy solution
satisfies all constraints in \divergent\ problem 
and has \linebreak $\totalcost(\configuration, \routing) \leq 
(1 + \beta) \totalcost(\configuration_{opt}, \routing_{opt})$. \eop
\end{theorem}

Note that this greedy scheme 
does not encounter the aforementioned problem
with $\costbip(q, r)$ not to be equal to
$cost(q,I_r)$.
Informally, the reason is due to the fact that the right hand-side 
of the inequality constraint in~\eqref{eq:div-greedy}
is a constant.
}

%% file: feature.tex
\section{\textsc{RITA}: Architecture and Functionality}
\label{sec:tool}

\begin{figure}[t]
        \begin{center}
                   \includegraphics[width=0.8\hsize]{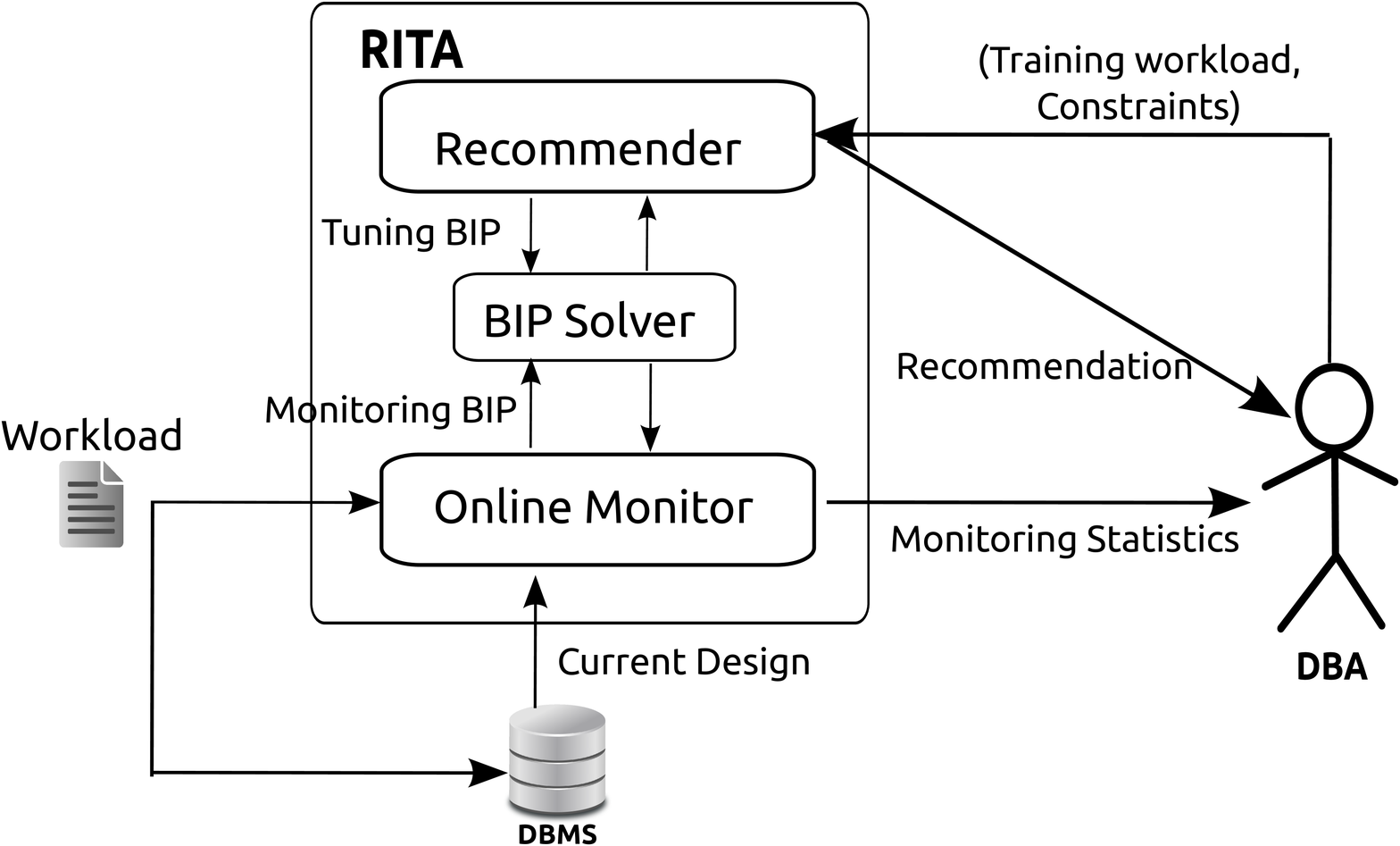}
                \caption{The architecture of RITA.}\label{fig:rita}
        \end{center} 
\end{figure}

In this section
we describe the architecture and the functionality of \rita, our proposed index-tuning
advisor. 
\rita\ builds on the reduction presented
in the previous section in order to offer a rich set of features.  

\eat{
 that leverages the theoretical result 
of reducing the divergent design tuning problem
to Binary Integer Programming 
in order to implement a unique set of features. 
}

Figure~\ref{fig:rita} shows the architecture of $\rita$. It comprises two main modules: the \textbf{online monitor}, which continuously analyzes the workload in order to detect changes and opportunities for retuning; and \textbf{the recommender}, which is invoked by the DBA in order to run a tuning session. As we will see later, both modules solve a variant of the $\divergent$ problem in order to perform their function. Also, both modules make use of the reduction we presented in the previous section in order to solve the respective tuning problems. For this purpose, they employ an off-the-shelf BIP solver. 
The remaining sections discuss the two modules in more detail. 

\subsection{Online Monitor}
\label{sec:rita-bip}

The online monitor maintains a divergent design
$(\configuration^\slide,\routing^\slide)$ that is continuously 
re-computed based on
the latest queries in the workload. Concretely, the monitor maintains a sliding window
over the current workload (the length of the window is a parameter defined by the DBA)
and then solves $\divergent$ using the sliding window as the training workload. Each
new statement in the running workload causes an update of the window and a
re-computation of $(\configuration^\slide,\routing^\slide)$.

Once computed, the up-to-date design $(\configuration^\slide,\routing^\slide)$ is compared against the current design $(\configuration^\curr,\routing^\curr)$ of the system, using the $\expected()$ metric of each design on the workload in the sliding window. The module outputs the difference between the two as the performance improvement if $(\configuration^\slide,\routing^\slide)$ were materialized. This output, which is essentially a time series since $(\configuration^\slide,\routing^\slide)$ is being continuously updated, can inform the DBA about the need to retune the system.

Clearly, it is important for the online monitor to maintain $(\configuration^\slide,\routing^\slide)$ up-to-date with the latest statements in the workload. For this purpose, the online monitor solves a bare-bones variant of $\divergent$ that assumes $\alpha=0$ (i.e., no failures) and does not employ any constraints except perhaps very basic ones (e.g., a space budget per replica). Beyond being fast to solve, this formulation also reflects the best-case potential to improve performance, which again can inform the DBA about the need to retune the system. $\rita$ allows the DBA to impose additional constraints inside the online monitor at the expense of taking longer to update the output of the online monitor.

\subsection{Recommender}

The DBA invokes the recommender module to run a tuning session, 
for the purpose of tuning the initial divergent design or 
retuning the current design when the workload changes. 
The DBA provides an instance of the $\divergent$ problem, e.g., a training workload, the parameter $\alpha$ and several constraints, and the recommender returns the corresponding (near-)optimal divergent design. The recommender leverages the BIP-based formulation of $\divergent$ in order to compute its output efficiently.

If desired by the DBA, the recommender can also return a set of possible designs that represent trade-off points within a multi-dimensional space. 
For example, suppose that 
the DBA specifies 
the workload-evaluation cost and the materialization cost 
of each design as the two dimensions of this space. 
We expect that 
a design 
with a higher materialization cost will have more indexes,  
and hence will have a 
lower workload-evaluation cost.
The recommender formulates a BIP to compute
an optimal divergent design that does not bound the materialization cost. 
The solution
provides an upper bound on materialization cost, 
henceforth denoted as $\Cmat$. 
Subsequently, the recommender formulates several tuning BIPs 
where each BIP puts a different threshold on 
the materialization cost based on $\Cmat$ and some factor 
(e.g., materialization cost should not exceed $.5\times\Cmat$). 
The thresholds for these Pareto-optimal designs 
can be predefined or chosen based on 
more involved strategies such as the Chord algorithm~\cite{DDY10}.
An important point is that the successive BIPs are essentially identical except for the modified constraint on the materialization cost, which enables the BIP solver to work fast by reusing previous computations.

The DBA can also add other parameters into this exploration.
For example,
adding the number of replicas as another parameter
will cause
the recommender 
to use the same process to generate designs for the hypothetical scenarios 
of expanding/shrinking the set of replicas. 
The final output can inform the DBA about the trade-off between workload-evaluation
cost and design-materialization cost, and how it is affected by the number of replicas. 

Besides being able to perform tuning sessions efficiently, $\rita$'s recommender module gains 
two important features through its reliance on a BIP solver. 

\begin{itemize}

	\item \textbf{Fast refinement.} As mentioned earlier, the BIP solver can reuse computation if the current BIP is sufficiently similar to previously solved BIPs. $\rita$ takes advantage of this feature to offer fast refinement of the solution for small changes to the input. E.g., the optimal divergent design can be updated very efficiently if the DBA wishes to change the set of candidate indexes or impose additional inter-replica constraints. 

	\item \textbf {Early termination.} In the course of solving a BIP, the solver maintains the currently-best solution along with a bound on its suboptimality. This information can be leveraged by $\rita$ to support early termination based on time or quality. For instance, the DBA may instruct the recommender to return the first solution that is within 5\% of the optimal, which can reduce substantially the total running time without compromising performance for the output divergent design. Or, the DBA may ask for the best solution that can be computed within a specific time interval. 
\end{itemize}

%% file: expt.tex
\section{Experimental Study}
\label{sec:expt}

This section presents the results of the experimental study 
that we conducted in order to evaluate the effectiveness of 
\rita. 
In what follows, we first discuss the experimental methodology and 
then present the findings of the experiments.

\begin{table}[tb]
{\small
 \centering

 \begin{tabular}{ l l } 
    \textbf{Parameter} & \textbf{Values} \\ \hline
  Number of replicas ($N$) & $2$, $\textbf{3}$, $4$, $5$ \\ 
  Routing multiplicity ($m$) & $1$, $\textbf{2}$, $3$ \\ 
  Space budget ($b$) & 0.25$\times$,  \textbf{0.5$\times$}, 1.0$\times$, INF \\ 
  Prob. of failure ($\alpha$) & \textbf{0.0}, $0.1$, $0.2$, $0.3$, $0.4$ \\
  Load skew ($\nodefactor$)	& $1.3$, $1.5$, $1.7$, $1.9$, $2.1$, \textbf{INF} \\ 
  Percentage-update ($\pupd$) 
	& $10^{-5}$, $10^{-4}$, \textbf{${\textbf{10}}^{\textbf{-3}}$}, $10^{-2}$ \\
   Sliding window ($w$) & $40$, $\textbf{60}$, $80$, $100$ \\
\hline
 \end{tabular}
\caption{Experimental parameters (default in bold).\label{tab:parameter}}
}
\end{table}

\begin{figure*}[tbh]
        \begin{tabular}{ccc}
        \includegraphics[width=0.3\textwidth]{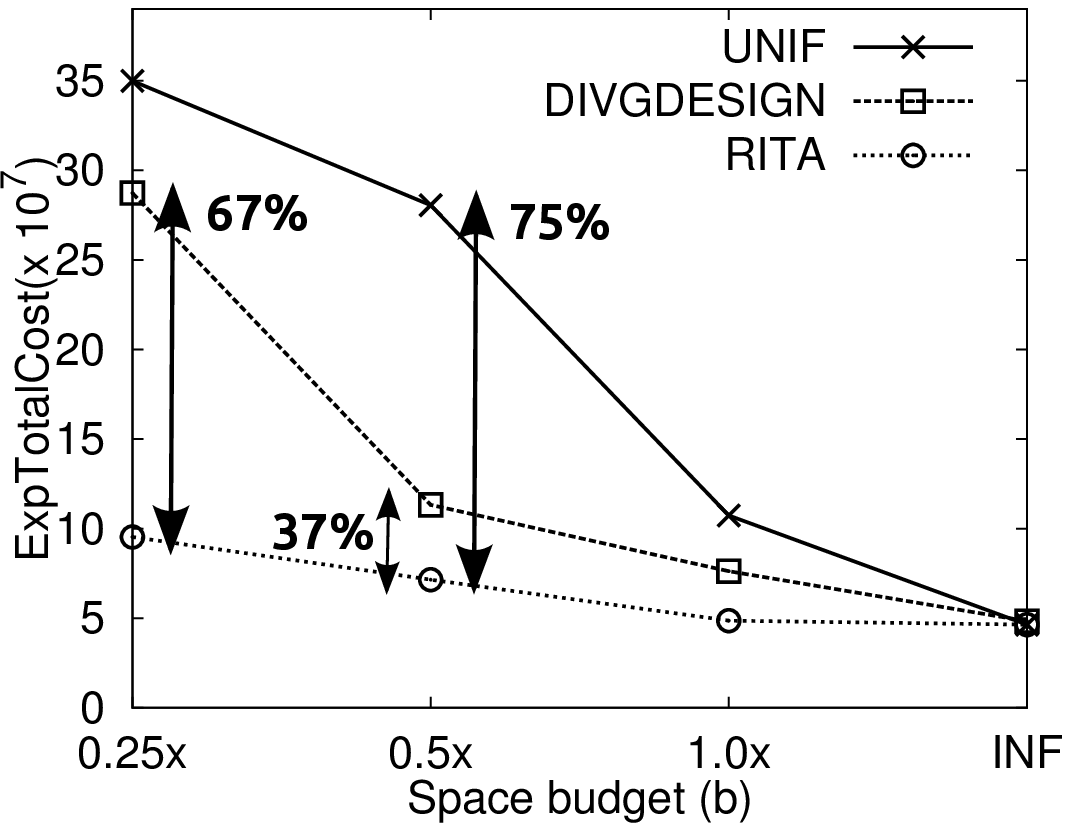} 
		& \hspace{-0.2cm} \includegraphics[width=0.3333\textwidth]{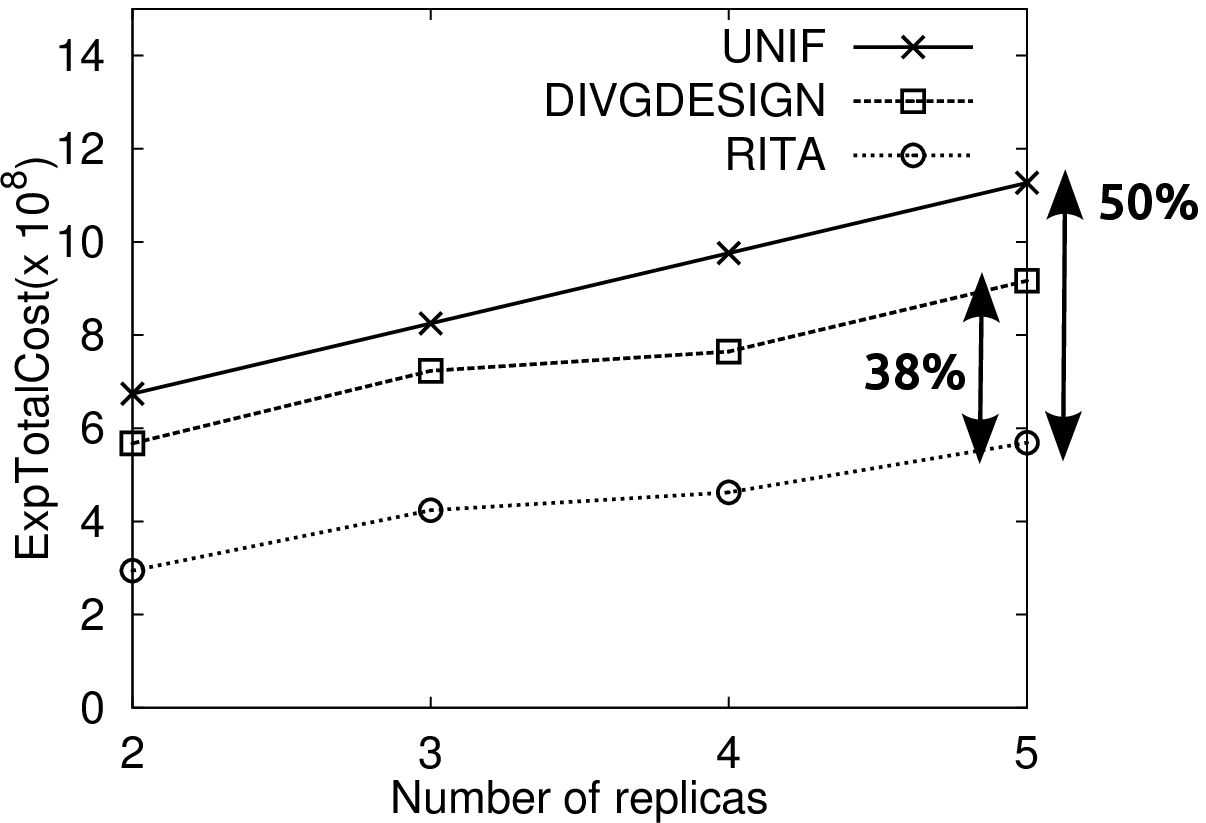} 
		& \hspace{-0.2cm} \includegraphics[width=0.33\textwidth]{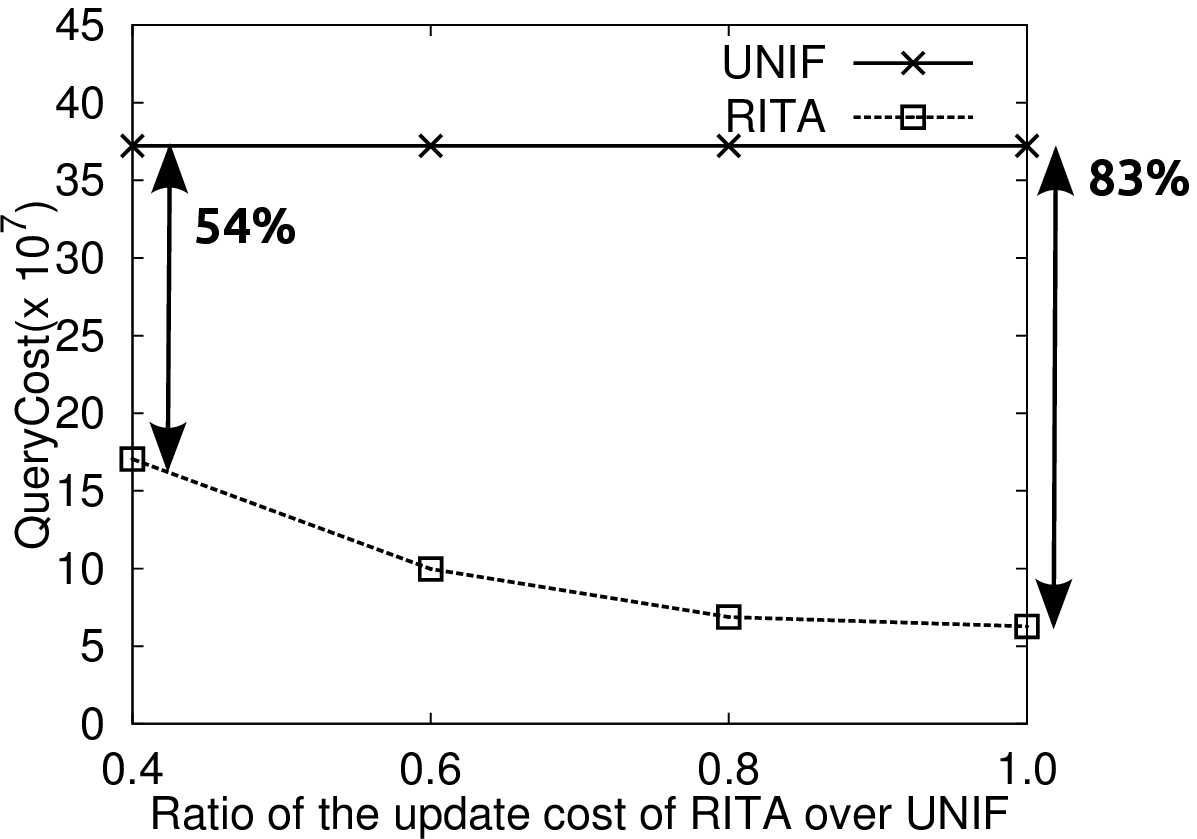} 
		\\
                \parbox{0.3\textwidth}{\cspace \centering\caption{Varying space budget
		on $\static$, 
                 $\alpha = 0$, $\nodefactor = +\infty$. 
		} \label{fig:equivalent}} 
		&	\parbox{0.3\textwidth}{\cspace \centering
		\caption{Varying number of replicas
		on $\mix$, 
                 $\alpha = 0$, $\nodefactor = +\infty$.}  \label{fig:update}} 
		&	\parbox{0.3\textwidth}{\cspace \centering
		\caption{Constraint the update cost
			on $\mix$, $\alpha = 0$, $\nodefactor = +\infty$.
			} \label{fig:update_constraint}} 
		\\
		\\
        \end{tabular}
        \vspace{-10pt}
\end{figure*}

\subsection{Methodology}

\stitle{Advisors.} Our experiments use a prototype implementation of $\rita$
written in Java. The prototype employs CPLEX v12.3 as the off-the-shelf BIP solver, and a custom implementation of INUM for fast what-if optimization. The database system in our experiments is the freely available IBM DB2 Express-C.
The CPLEX solver is tuned to return the first solution
that is within 5\% of the optimal. In all experiments, we use $\prita$ to denote the divergent design computed by $\rita$. 

We compare \rita\ against the heuristic advisor \divgdesign\ that was 
introduced in the original study of divergent designs~\cite{Divergent2012}. 
$\divgdesign$ employs IBM's physical design advisor
internally. Similar to~\cite{Divergent2012}, we run \divgdesign\ 
five times and output the lowest-cost design out of 
all the independent runs. We denote this final design as $\pdd$.
We note that the comparison against $\divgdesign$ concerns only a restricted definition of the general tuning problem, since $\divgdesign$ supports only a space budget constraint and does not take into account replica failures. 

We also include in the comparison the common practice of using the same index configuration with each replica. The identical configuration is computed by invoking 
the DB2 index-tuning advisor on the whole workload. We use $\punif$ to refer to the resulting design.

\stitle{Data Sets and Workloads.}
We use a 100GB TPC-DS database~\cite{tpcds} for our experiments, along with three different workloads, namely $\static$, $\mix$ and $\dynamic$. $\static$ comprises 40 complex 
TPC-DS \add{benchmark} queries 
that are currently supported by our INUM implementation~\cite{inumplus}.
$\mix$ adds INSERT statements that model updates to the base data. $\dynamic$ models a workload of 600 queries that goes through three phases, each phase corresponding to a specific distribution of the queries that appear in $\static$. The first phase corresponds mostly to queries of 
low execution cost\footnote{The execution cost
is measured with respect to the optimal index-set
for each query returned by the DB2 advisor.}, 
then the distribution is inverted for the second phase, and reverts back to the starting distribution in the first phase.

In all cases, the weight for each query is set to one, whereas the update of each INSERT statement is determined as the product of the cardinality of the corresponding relation and a 
\emph{percentage-update} parameter ($\pupd$). 
This parameter allows us to simulate different volumes of updates when we test the advisors. 

\eat{
Note that the size of the database  does not affect
the trends observed in our experiments, 
as all performance metrics are based on the DB2 optimizer's cost model (more on this later).
}

\stitle{Candidate Index Generation.}
Recall from Section~\ref{sec:problem} that 
the \divergent\ problem
assumes that a set of candidate indexes $\Iset$ is provided as input. 
There are many methods for generating $\Iset$ based on the
database and representative workload. In our setting, 
we use DB2's service to select the optimal indexes per query (without any space constraints) 
and then perform a union of the returned index-sets. 
The resulting index-set, which is optimal 
for the workload in the absence of constraints and update statements, 
contains $103$ candidate indexes and has a total size of $265$GB.

\stitle{Experimental Parameters.} 
Our experiments vary the following parameters: 
the number of replicas $N$, 
the per-replica space budget $b$, 
the probability of failure $\alpha$,
the load-skew factor $\nodefactor$, 
the percentage of updates in the workload $\pupd$ (for $\mix$), 
and the size of the sliding window $w$ for online monitoring. 
The routing multiplicity factor ($m$) 
is set to be $\lceil N/2 \rceil$. 
We report the additional experimental results
when varying $m$ in Appendix~\ref{app:vary-m}.
Table~\ref{tab:parameter} shows the parameter values 
tested in our experiments.  
Note that the storage space budget is measured as a multiple of the base data size, 
i.e., given TPCDS $100$ GB base data size, a space budget of $0.5\times$ indicates 
a $50$ GB storage space budget.

\stitle{Metrics.} We use $\expected()$ to measure the performance of a divergent design. To allow meaningful comparisons among the designs generated by different advisors, we compute this metric for a specific design by invoking DB2's what-if optimizer for all the required cost factors. This methodology, which is consistent with previous studies on physical design tuning, allows us to gauge the effectiveness of the divergent design in isolation 
from any estimation errors in the optimizer's cost models.
In some cases, we also report the performance improvement of 
$\prita$ over $\pdd$ and $\punif$, where
the performance improvement of a design $X$
over a design $Y$ is computed as 
$1 - \expected(X) / \expected(Y)
$. We also report the time that is taken to execute the index advisor 
for the corresponding divergent design.

\stitle{Testing Platform.} 
All measurements are taken on a machine running 64-bit Ubuntu OS
with four-core Intel(R) Core(TM) i7 CPU Q820 @1.73GHz CPU
and 4GB RAM.

\subsection{Results}
\label{sec:static}

\stitle{Basic Tuning Problem.}
We first consider a basic case of \divergent\ when $\alpha = 0$ and $\nodefactor = +\infty$, i.e., no failures occur and there is no constraint on load skew. There is a single constraint on the divergent design which is the per-replica space budget. 
This setting corresponds essentially to the original problem statement in~\cite{Divergent2012}.

We begin with a set of experiments that evaluates the performance of $\rita$ and the competitor advisors on the query-only workload $\static$. In this case indexes can only bring benefit to queries, and hence the only restraint in materializing indexes comes from any constraints. Figure~\ref{fig:equivalent}
shows the performance of the divergent designs computed by 
\rita, \divgdesign, and \unif, as we vary the space budget parameter. 
(All other parameters are set to their default values according to Table~\ref{tab:parameter}.) 
The results show that \rita\ consistently outperforms 
the other two competitors for a wide range of space budgets. 
The improvement is up to 75\% over $\punif$ 
and up to 67\% for $\pdd$,
i.e., the performance of $\prita$ is 4$\times$ better than $\punif$
and is 3$\times$ better than $\pdd$.
Another way to view these results is that $\rita$ can make much 
more effective usage of the aggregate disk space for indexes. 
For instance, $\prita$ at $b=0.25\times$ matches the performance of 
$\pdd$ at $b=1.0\times$, i.e., with four times as much space for indexes. 
In all cases, $\rita$'s better performance can be attributed to the fact that it searches a considerably larger space of possible designs, through the reduction to a BIP. As the space budget increases, the performance of $\pdivg$, $\pdd$ and $\punif$ converge as all beneficial indexes can be materialized in every design.

We next examine the performance of $\rita$ and 
the competitor advisors on a workload of queries and updates.
Figure~\ref{fig:update} reports the performance of $\pdivg$, $\pdd$
and $\punif$ for the workload $\mix$, as we vary the number of replicas in the system. 
We chose this parameter as updates have to be routed to all replicas 
and hence it controls directly the total cost of updates. 
We observe that the improvement of \rita\ over \unif\ is 
in the order of $50\%$ and the improvement of \rita\ over \divgdesign\ is $38\%$. 
Not surprisingly, the improvements increase with the number of replicas. 
The reason is that 
$\rita$ is able to find designs with much fewer indexes per replica 
compared to $\punif$ and $\pdd$, which contributes to a lower update cost. 
For instance for $N = 3$ and $b = 0.5 \times$,
the number of indexes per replica of $\prita$ is $(44, 44, 31)$
compared to $(70, 70, 70)$ for $\punif$ and $(46, 50, 53)$ for $\pdd$. 
We conducted similar experiments with different weights 
for the update statements and observed similar trends.

\begin{figure*}[tbh]
        \begin{tabular}{ccc}
        \includegraphics[width=0.32\textwidth]{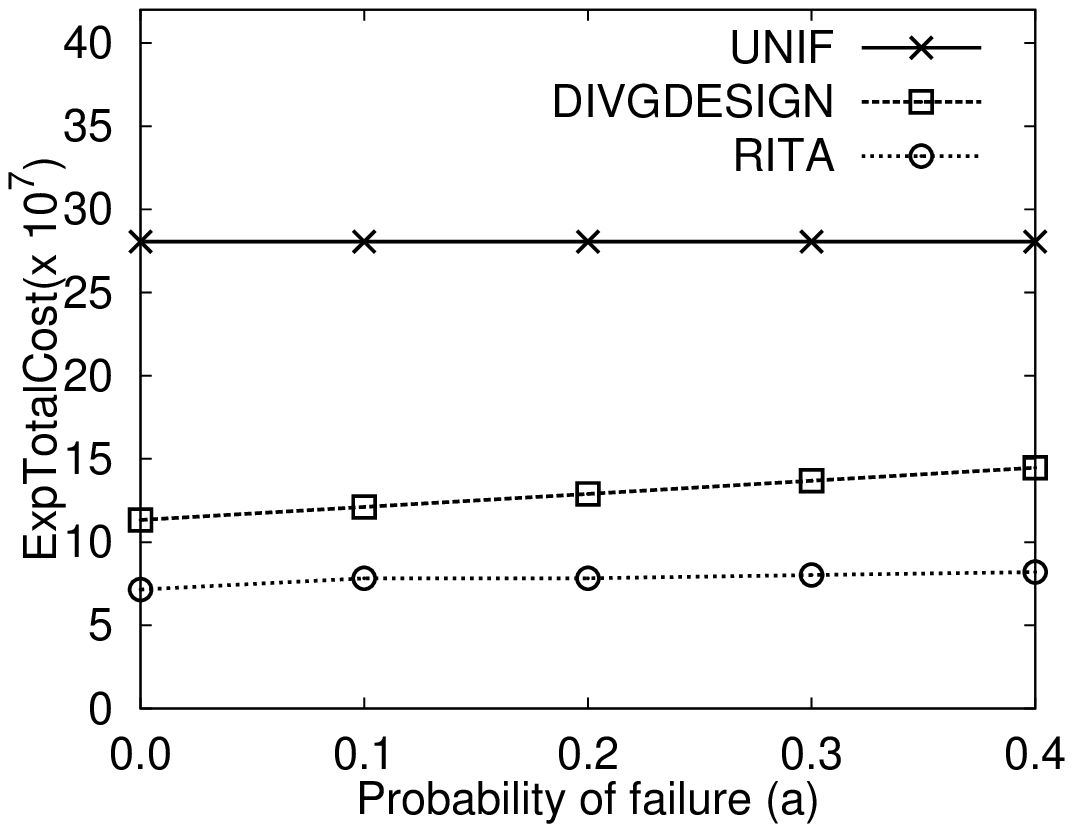} 
		& \hspace{-0.2cm} \includegraphics[width=0.32\textwidth]{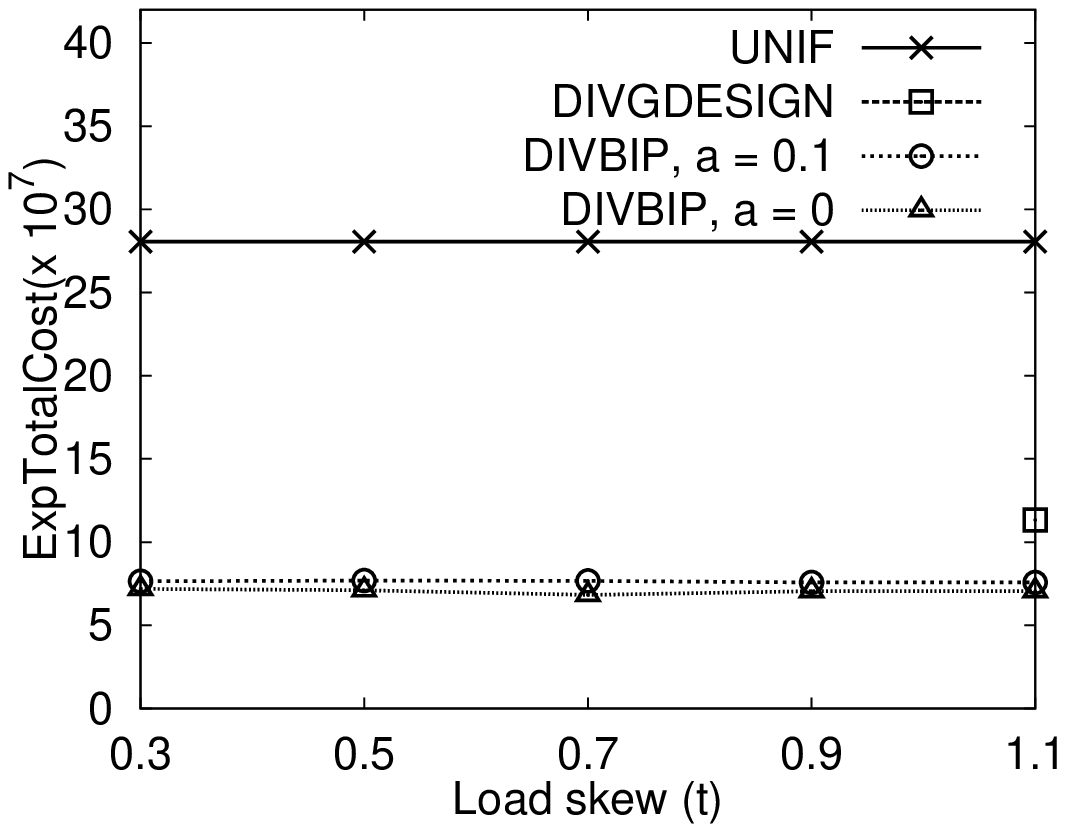} 
		& \hspace{-0.2cm} \includegraphics[width=0.32\textwidth]{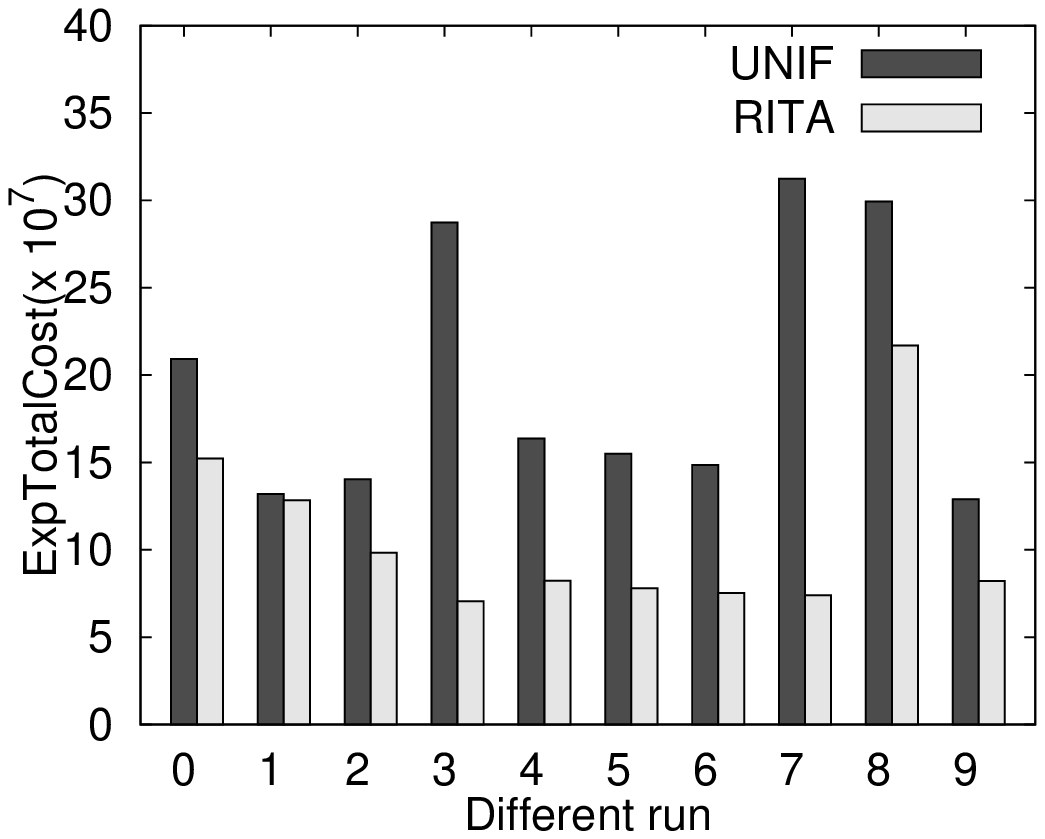} 
		\\
		\parbox{0.3\textwidth}{\cspace \centering\caption{ Varying probability of failure on 
			$\static$, $\alpha \geq 0$, $\nodefactor = +\infty$. } \label{fig:failure}} 
		&	\parbox{0.3\textwidth}{\cspace \centering\caption{ Varying load skew on 
			$\static$, $\alpha \geq 0$, $\nodefactor < +\infty$. } \label{fig:imbalance}} 
		&	\parbox{0.3\textwidth}{\cspace \centering\caption{Routing 
                queries} \label{fig:unseen}} 
		\\
		\\
        \end{tabular}
        \vspace{-10pt}
\end{figure*}

The next experiment examines how $\rita$'s advanced functionality can control even further the cost of updates. Instead of having $\rita$ minimize the combined cost of queries and updates, we instruct the advisor to perform the following constrained optimization: minimize query cost such that update cost is at most $x\%$ of the update cost of a uniform design. Essentially, the desire is to make updates much faster compared to the uniform design, and also try to get some benefits for query processing. This changed optimization requires minimal changes to the underlying BIP: the objective function includes only the cost of evaluating queries, and the constraints include an additional linear constraint on the total update cost based on the update cost of the uniform design (which can be treated as a constant). The ease by which we can support this advanced functionality reflects the power of expressing $\divergent$ as a BIP.\tightpara

Figure~\ref{fig:update_constraint} depicts the cost of the query workload under $\prita$ as we vary the factor that bounds the update cost relative to $\punif$. For comparison we also show the cost of the query workload for $\punif$. 
The results show clearly that the designs computed by $\rita$ 
can improve performance dramatically even in this scenario. 
As a concrete data point, when the bounding factor is set to $0.4$, 
$\prita$ makes query evaluation more than 2$\times$ cheaper compared to $\punif$ and 
incurs an update cost that is less than half the update cost of $\punif$. 

Overall, our results demonstrate that $\rita$ clearly outperforms its competitors on the basic definition of the divergent-design tuning problem. From this point onward, we will evaluate $\rita$'s effectiveness with respect to the generalized version of the problem (i.e., including failures and a richer set of constraints). In the interest of space, we present results with query-only workloads, as the trends were very similar when we experimented with mixed workloads.\tightpara

\stitle{Factoring Failures.}
We first evaluate how well $\rita$ can tailor the divergent design in order to account for possible failures, as captured by the failure probability $\alpha$. 

Figure~\ref{fig:failure} shows the $\expected()$ metric for $\prita$, $\pdd$ and
$\punif$ as we vary the probability of failure $\alpha$. There are two interesting
take-away points from the results. The first is that $\prita$ has a relatively stable
performance as we vary $\alpha$. Essentially, we can reap the benefits of divergent
designs even when there is an increased probability of failure in the system, as long
as there is a judicious specialization for each replica and a controlled strategy to
redistribute the workload (two things that $\rita$ clearly achieves). The second
interesting point is that the gap between $\prita$ and $\pdd$ increases with $\alpha$.
Basically, $\pdd$ ignores the possibility of failures (i.e., it always assumes that
$\alpha=0$) and hence the computed design $\pdd$ cannot handle effectively a
redistribution of the workload when a replica becomes unavailable. As a side note, the
cost of $\punif$ is unchanged for different values of $\alpha$, since each query has
the same cost under $\punif$ on all replicas, and hence a redistribution of the
workload does not change the total cost.

\stitle{Bounding Load Skew.}
We next study how $\rita$ handles a (inter-replica) constraint on load skew. 
Recall that the constraint has the following form: for any two replicas, 
their load should not differ by a factor of more than $1+\tau$, 
where $\tau\ge 0$ is the load-skew parameter. 
A balanced load distribution is important for good performance in a distributed system and hence we are interested in small values for $\tau$. The ability to satisfy such constraints is part of $\rita$'s novel functionality. 

Figure~\ref{fig:imbalance} shows the performance of $\pdivg$, $\pdd$ and $\punif$
as we vary parameter $\nodefactor$ that bounds the load skew 
(recall that $\nodefactor=0$ implies no skew). 
We report two sets of results for $\rita$ corresponding to $\alpha=0$ (no failures) 
and $\alpha=0.1$ (10\% chance that one replica will fail) respectively, 
in order to examine the interplay between $\alpha$ and $\nodefactor$. 
Note that we report the results for the greedy version of $\rita$, 
which are identical to the exact solution of the constraint. 
The chart shows a single point corresponding to $\pdd$, 
given that it is not possible to constrain load skew within $\divgdesign$. 
As shown, $\pdd$ has a significant load skew of up to a $2x$ difference between replicas. 
This magnitude of skew limits severely the ability of the system to maintain a balanced load and to route queries effectively. 
In contrast, $\rita$ is able to compute designs that maintain a low expected cost (up to 4$\times$ better than $\unif$) and also satisfy the bound on load skew. These savings are not affected by the value of $\alpha$--$\rita$ is again able to make a judicious choice for the divergent design in order to satisfy all constraints and handle failures. 
Note that the uniform design trivially satisfies the load-skew constraint for all values of $\nodefactor$ as every replica has the same design and hence the system can be perfectly balanced.

\begin{table}[tb]
   \centering
        \begin{tabular}{cc}
                \begin{tabular}{|c||c|c| }  \hline
                                     & $\alpha = 0$ & $\alpha > 0$ \\ \hline \hline
                $\nodefactor = +\infty$      & $4$  & $60$ \\ \hline
                $\nodefactor < +\infty$      & $9$  & $84$ \\ \hline
                \end{tabular}
                &
                \begin{tabular}{|c||c|c| }  \hline
                                     & $\alpha = 0$ & $\alpha > 0$ \\ \hline \hline
                $\nodefactor = +\infty$      & $20$  & $120$ \\ \hline
                $\nodefactor < +\infty$      & $30$  & $146$ \\ \hline
                \end{tabular}
        \\

        (a) Workload $\static$ & (b) Workload $\dynamic$
        \end{tabular}
\caption{The average running time of \rita\
(in seconds)}
\label{tab:time}

\end{table}

\stitle{Running Time.}
Given an instance of the basic \divergent\ problem ($\alpha = 0$, $\nodefactor = +\infty$), \rita\ spends $180$ seconds to initialize INUM, a step that is dependent solely on the input workload, and then requires only four seconds to formulate and solve the resulting BIP. An important point is that the initialization step can be reused for free 
if the workload remains unchanged, e.g., if the DBA runs several tuning sessions using the same workload but different constraints each time. Each subsequent tuning session can thus be executed in the order of a few seconds, offering an almost interactive response to the DBA. 

Table~\ref{tab:time}(a) shows the running time for $\rita$ 
on $\static$ workload
as we vary the load-skew factor and the probability of failure, 
two parameters that correspond to novel features of our generalized tuning problem. Note that the time to initialize INUM remains the same as before and is excluded from all the cells of the table. 
Clearly, the new features complicate the tuning problem and hence have an impact on running time. 
Still, even for the most complex combination ($\nodefactor > 0$ and $\alpha>0$) $\rita$ 
has a reasonable running time of at most $84$ seconds. 
Moreover, as noted in Section~\ref{sec:tool}, $\rita$ can always be invoked with a time threshold and 
return the best design that has been identified within the allotted time.

Table~\ref{tab:time}(b)
shows
the same details about the running time of $\rita$
on $\dynamic$ workload, consisting of $600$ queries.
$\rita$ also runs efficiently for this large workload.

\begin{figure*}[tbh]
        \begin{tabular}{ccc}
        \includegraphics[width=0.32\textwidth]{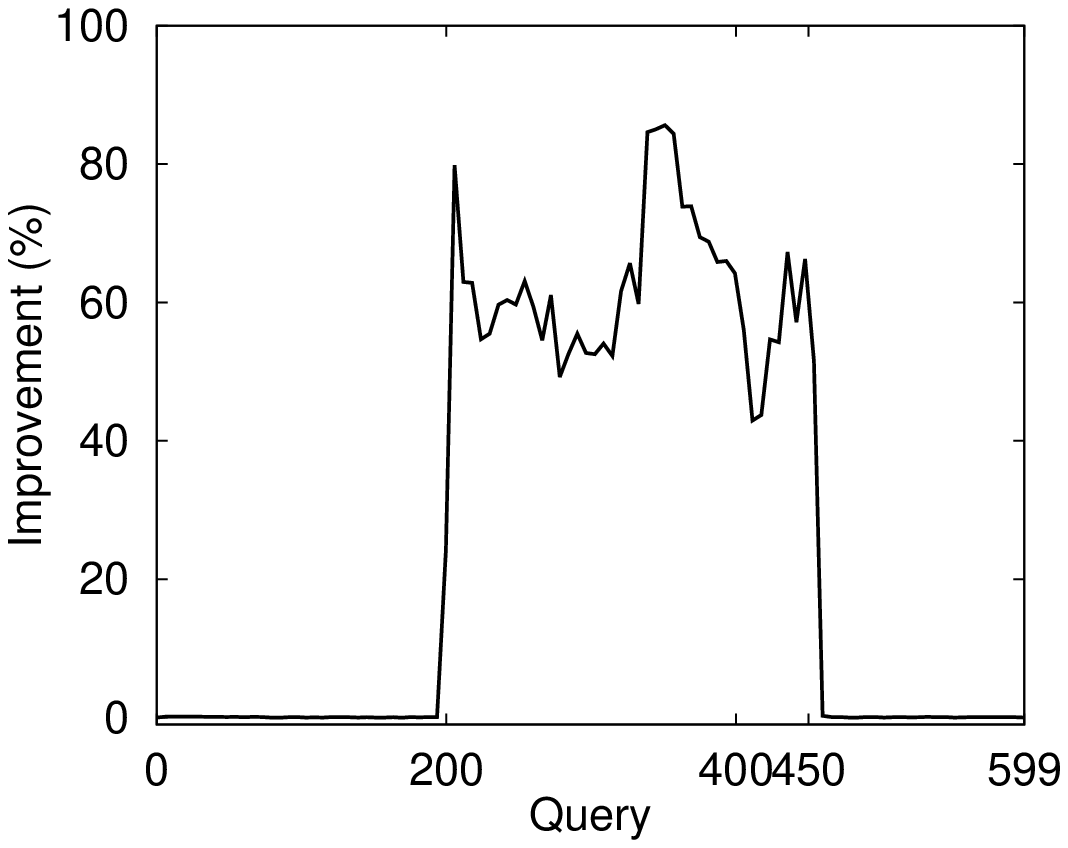} 
		& \hspace{-0.2cm} \includegraphics[width=0.3\textwidth]{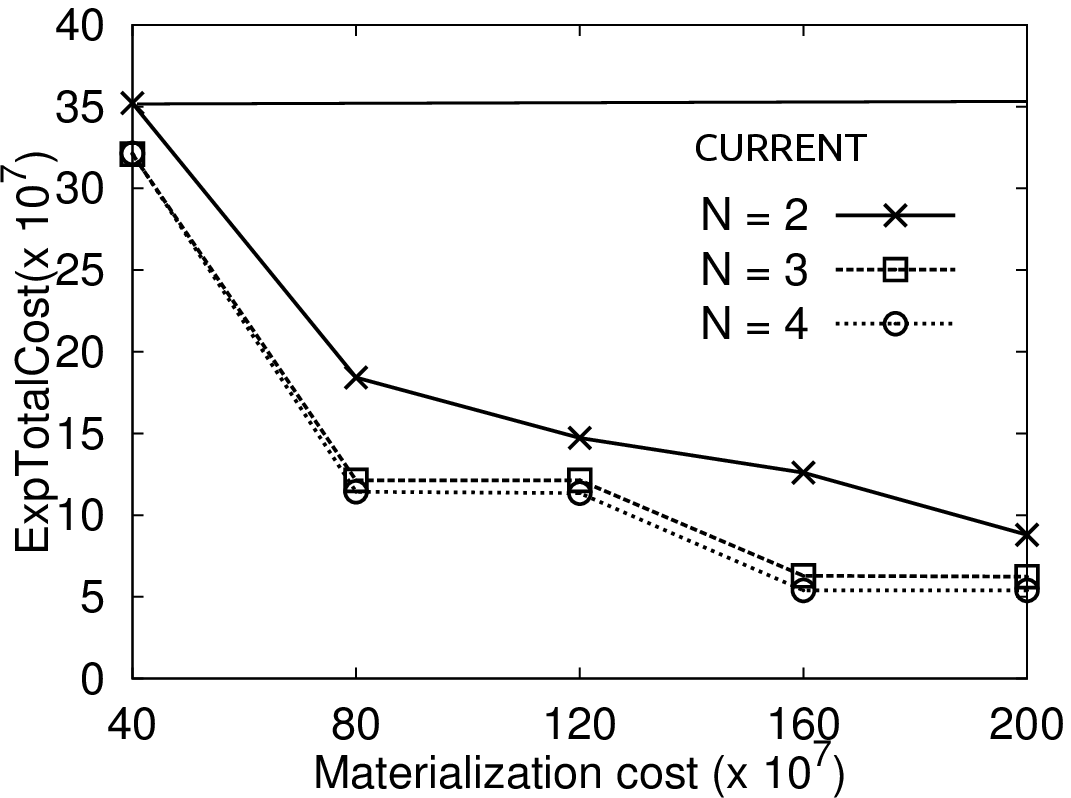}
		& \hspace{-0.2cm}  
		\includegraphics[width=0.3\textwidth]{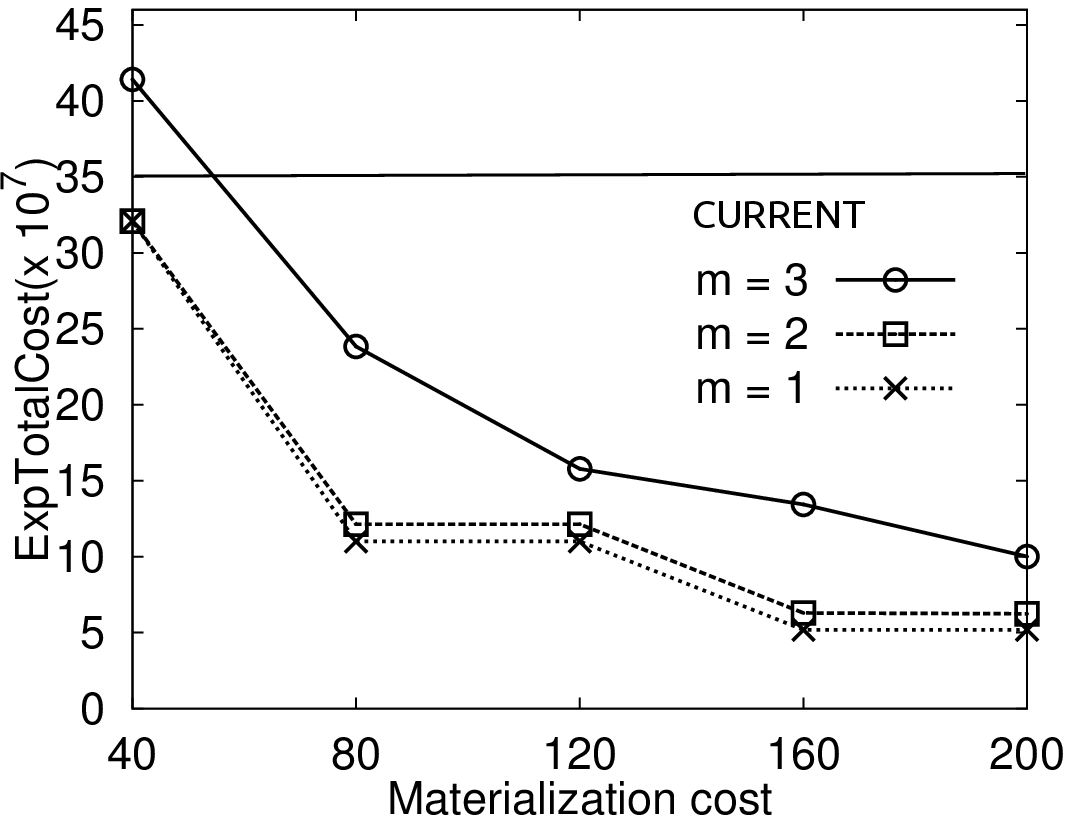} 
		\\
          \\     
		\parbox{0.3\textwidth}{\cspace \centering\caption{Online monitoring} \label{fig:online}} 
		&	\parbox{0.3\textwidth}{\cspace\centering\caption{Elasticity retuning,
		varying number of replicas and materialization costs} 
	\label{fig:elastic}} 
	& 
		\parbox{0.3\textwidth}{\cspace \centering\caption{Elasticity retuning,
		varying routing multiplicity factor and materialization costs} \label{fig:elastic_m}} 
		\\
		\\
        \end{tabular}
        \vspace{-10pt}
\end{figure*}

\stitle{Routing.}
The next set of experiments examines the effectiveness of the routing scheme we introduced in Section~\ref{sec:routing}, 
which determines how to route unseen queries (i.e., queries not in $W$ for which the routing functions $\route{j}$ cannot be applied) 
to ``good'' specialized replicas.

Our test methodology splits $\static$ 
into two (sub)workloads: (1) a training workload that plays the role of $W$ and 
consists of $30$ randomly-chosen queries of $\static$,
and (2) a testing workload that plays the role of the unseen queries and 
consists of the remaining $10$
queries. 
We compute a divergent design $\pdivg$ for the training workload,
and 
route
the queries in $\static$ (including
both seen and unseen queries)
assuming $\pdivg$ is deployed. 
For comparison, we apply the same methodology to the uniform design: 
we first derive $\punif$
for the training workload and then 
route the queries in $\static$ workload
in round-robin fashion. 
We repeat this experiment for ten independent runs, where each run involves a different random split of the workload. 

Figure~\ref{fig:unseen} shows the expected cost of the workload for $\prita$ and $\punif$ for each run. 
The results show that \rita\ outperforms \unif\ consistently, even though replica specialization does not take into account the unseen queries. The improvements vary across different runs depending on the choice of the workload split, but overall we can reap the benefits of divergent designs even with incomplete knowledge of the workload. 

\eat{
\rita\ takes $1.5$ seconds to compute the routing functions
for an unseen query in the testing workload (Note that it is nearly for free to 
compute the routing functions for 
the seen queries in the training workload.) 
}

\stitle{Online Monitoring.} The aforementioned routing scheme can help the system cope with unseen queries, but at some point it may become necessary to retune the divergent design if the actual workload is substantially different than the training workload. The next experiment evaluates the online-monitoring module inside $\rita$ which is designed for the task of detecting workload changes. 

We assume that the system receives the dynamic workload $\dynamic$, which shifts to a different query distribution after query 200 and then shifts back to the original distribution at query 400. Initially, the system is equipped with a divergent design $\pcurr$ that is tuned with a training workload from the first query distribution. The monitoring module continuously computes a divergent design $\pslide$ based on a sliding window of the last 60 queries in the workload, and outputs the improvement on $\expected()$ if $\pslide$ were used instead of $\pcurr$. 

Figure~\ref{fig:online}
shows the monitoring statistics produced by the online-monitoring module of $\rita$ for the $\dynamic$ workload. Matching our intuition, the output shows that $\pslide$ has small improvements for the first 200 queries (around 30\%), 
since the current design $\pcurr$ is already tuned for the particular phase of the workload. 
However, as soon as the workload shifts to a different distribution, the output shows a considerable improvement of more than 60\%. This can be viewed as a  strong indication that a retuning of the system can yield significant performance improvements. The spike tapers off close to query $450$, 
since in this experiment the workload shifts back to its previous distribution and hence there is no benefit to changing the current design.

$\rita$ requires $1.2$ seconds on average to analyze each new query in this workload, and can thus generate an output that accurately reflects the actual workload. 
We conducted similar experiments with different values of the 
length of the sliding window and observed similar results. For instance, \rita\ takes at most $3$ seconds when the sliding window is set to $100$ statements.

\stitle{Elastic Retuning.} After observing the monitoring output, 
the DBA can invoke the recommender module to examine different recommendations 
for retuning the system in an elastic fashion. 
The next set of experiments evaluate how fast $\rita$ can generate 
these recommendations and also their quality.

We employ a scenario that builds on the previous experiment on online monitoring. 
Specifically, we assume that the DBA invokes the recommender using the sliding window of 60 queries 
that corresponds to the spike in Figure~\ref{fig:online}. 
Moreover, the DBA specifies two dimensions of interest with respect to a new divergent design: 
the workload-evaluation cost and the cost of materializing the design.
Also, the DBA wants to study the effect of shrinking and expanding
the number of replicas.
We assume that the DBA sets the probability of failure 
($\alpha$) to be $0$ 
in order to allow \rita\ to execute fast and generate the output in a timely fashion. 
After inspecting the output, the DBA may invoke another (more expensive) 
tuning session for a specific choice of replicas (or routing multiplicity factor) 
and reconfiguration cost, 
and a non-zero $\alpha$. 
Our results in Figure~\ref{fig:failure} show that \rita\ 
can compute a divergent design that matches the same level of performance as the case for $\alpha=0$.

Figures~\ref{fig:elastic}
shows the output of the recommender based on our testing scenario.
Each point $(x,y)$ on the chart corresponds to a divergent design that requires $x$ 
cost units to materialize and whose $\expected()$ is equal to $y$. 
The three curves labeled $N=z$, $z \in \{2,3, 4\}$, 
represent divergent designs that employ $z$ replicas. 
We assume that $N=3$ is the current setting in the system, 
and hence $N=2$ (resp. $N=4$) represents dropping (resp. adding) a replica. 
The chart also shows the $\expected()$ metric of the current design, for comparison. 
As shown, there are several options to significantly improve (by up to $7\times$) 
the performance of the current design. Moreover, the DBA obtains the following valuable information:
there is a least materialization cost in order to get some improvement;
designs that require more than 160 units of materialization cost
offer diminishing returns for $N = 3$ and $N = 4$;
and there is not much benefit to increasing the number of replicas, 
since $N=3$ and $N=4$ have virtually identical performance. 
Based on these data points, the DBA can make an informed decision about how to retune the divergent design in the system. 
$\rita$ requires a total of $20$ seconds to generate the points in the chart. 
Note that the recommender does not have to initialize INUM for the training workload, 
as this initialization has already been performed inside the monitoring module.
This short computation time facilitates
an exploratory approach to index tuning.

We employ another scenario that is similar
to the previous one except that 
we assume the DBA wants to study
the effect of using different values for the routing multiplicity factor, 
while keeping the number of replicas unchanged.

Figure~\ref{fig:elastic_m}
shows the output of the recommender based on the above testing scenario.
The three curves labeled $m=z$, $z \in \{1, 2, 3\}$,
represent divergent designs that have the routing multiplicity factor $z$
(We assume that $m=2$ is the current setting in the system).
We observe that designs that require more than 80 units of materialization cost 
when routing queries for $m = 2$
has slightly better performance 
when routing queries for $m = 1$. 
This result indicates that
we can obtain designs with some flexibility 
in routing queries (i.e., $m = 2$)
and without sacrifying much in terms of performance as designs that have
the most specialization (i.e., $m = 1$). 
$\rita$ requires a total of $10$ seconds to generate the points in the chart. 


%% file: conclusion.tex
\section{Conclusion}

In this paper, we introduced \rita, a novel
index tuning advisor for replicated databases, 
that provides DBAs with a powerful tool for divergent
index tuning. The key technical contribution of $\rita$
is a reduction of the problem to
a compact binary integer program,
which enables the efficient computation of a (near-)optimal divergent design using mature, off-the-shelf software for linear optimization.
Our experimental studies 
demonstrate that, compared to state-of-the-art solutions,
\rita\ offers richer tuning functionality
and is able to compute divergent designs
that result in significantly better performance.

%% file: appendix.tex
\appendix

\section{Proving Theorem 1 }
\label{app:hard}

We reduce the original problem studied in~\cite{Divergent2012} to $\divergent$ by proving their equivalence when $\alpha=0$ and $C$ contains solely a space-budget constraint per replica. Since the original problem is NP-Hard, the same follows for $\divergent$.  The result in Lemma~\ref{lemma:equivalance} (See below) is the key to prove their equivalence. 
It is important to note from Section~\ref{sec:compare} 
that in the general setting of \divergent, 
$\rcomponent{0}(q)$ might not correspond to the $m$ replicas 
with the least evaluation cost for $q$.

\begin{lemma}\label{lemma:equivalance}
	In the problem setting of $\divergent$ when $\alpha=0$ and $C$ contains solely a space-budget constraint per replica,
	$\rcomponent{0}(q)$ corresponds to the $m$ replicas 
with the least evaluation cost for $q$.	
\end{lemma}

We prove Lemma~\ref{lemma:equivalance}
using contradiction. 
Assume that for some query $q$, there exist two replicas $r_1$ and $r_2$ such that
$r_1 \in \rcomponent{0}(q)$, $r_2 \not\in \rcomponent{0} (q)$
and $cost(q, I_{r_1}) > cost(q, I_{r_2})$.
We then derive another routing function $\routing'$
that is similar to $\routing$ except that  $\rcomponent{0}'$ is slightly modified as follows: 
$\rcomponent{0}'(q) = \rcomponent{0}(q) \cup \{ r_2 \} - \{ r_1 \} $. 
Clearly, $\totalcost(\configuration, \routing) > \totalcost(\configuration, \routing')$.
This contradicts to the requirement to minimize $\totalcost(\configuration, \routing)$
in the problem setting of \divergent.

\section{Proving Theorem 2}
\label{app:div-equivalent}

We prove the theorem in two steps. 
First, we show that every divergent design $(\configuration, \routing)$
corresponds to a value-assignment $\va$
for variables in the BIP 
such that $\va$ satisfies the constraints
(Lemma~\ref{lemma:space}). 
This property guarantees that
the solution space of the BIP contains all possible solutions for the
divergent design tuning problem. 
Subsequently, we prove that the optimal assignment $\va^*$ corresponds
to a divergent design.
Combining these two results, we can then conclude the
correctness of the theorem (Lemma~\ref{lemma:optimal}).

To simplify the presentation and without loss of generality,
we prove the theorem for the basic \divergent\
when $\alpha = 0$, $C = \emptyset$ and
the workload comprises solely queries, i.e., $\workload = Q$.

Given a valid-assignment $\va$, we use $BIPcost(\va)$ to denote the value
of the objective function of the BIP under the assignment $\va$.

\begin{lemma} \label{lemma:space}
For any divergent physical design $(\configuration, \routing)$,  
there is an assignment $\va$ s.t.
$\totalcost(\configuration, \routing) = BIPcost(\va)$. 
\end{lemma}

\begin{lemma} \label{lemma:optimal}
	Let $\vastar$ denote the solution to the BIP problem. 
	Then, $\totalcost(\configuration, \routing) = BIPcost(\vastar)$,
	where $(\configuration, \routing)$ is the divergent design derived 
	from $\vastar$. 
\end{lemma}

\subsection{Proof of Lemma~\ref{lemma:space}}

Given a divergent design $(\configuration, \routing)$
and for every query $q \in Q$, 
using the linear 
decomposability property, we can express
the cost of $q$ at replica $r \in \route{0}(q)$
as:
\[
cost(q, \confcomponent{r}) = \beta_{p} + \sum_{i \in [1,n], a = Y[i]} \gamma_{pa}
\]
for some choice of $p = p^r \in \TPlans{q}$ and 
$Y = Y^{p, r} \in \atom(\confcomponent{r})$. 
We assign
the values for variables as follows.

\begin{itemize}
	\item $\va(t^r_q) = 1$ if $r \in \route{0}(q)$, 
	\item $\va(y^r_{p})= 1$ if $p = p^r$, $r \in \route{0}(q)$, 
	\item $\va(x^{r}_{pa}) = 1$ if $p = p^r$, $r \in \route{0}(q)$ and $a = Y^{p, r}[i]$, $i \in [1,n]$, 
	\item $\va(s^{r}_a) =1$ if $a \in I_r$,  $r \in [1, N]$, and 
	\item The other cases of variables are assigned value $0$
\end{itemize}

We observe that under this assignment, 
all constraints in the BIP are satisfied. 
For instance, since $\va(t^r_q) = 1$ when $r \in \route{0}(q)$ and $\route{0}(q)$
has $m$ values,
it can be immediately derived that $\sum_{r \in [1,N]} t^r_q = m$, i.e., 
constraint~\eqref{eq:div-top-m-q} is satisfied.

By eliminating terms with value $0$, we obtain the following results. 
\begin{equation*}
BIPCost(\va) = \sum_{q \in Q} \sum_{r \in \route{0}(q) } \frac{f(q)}{m} \costbip(q, r)
\end{equation*}
\begin{equation*}
\costbip(q, r) =  \beta_{p} + \sum_{i \in [1,n], a = Y[i]} \gamma_{pa},  \mbox{ for } r \in \route{0}(q), p = p^r, 
Y = Y^{p, r} \in \atom(\confcomponent{r})
\end{equation*}

Thus, $BIPCost(\va) = \totalcost(\configuration, \routing)$.

\subsection{Proof of Lemma~\ref{lemma:optimal} }

The following arguments are derived based on the assumption that
$\vastar$ satisfies the BIP formulation.

First, based on~\eqref{eq:div-top-m-q},
we derive that for every query $q$, there exists a set $S_q = \{ r \ | \ r \in [1, N] \}$ and $|S_q| = m$
such that $\vastar(t^r_q) = 1$ iff $r \in S_q$.

Second, based on~\eqref{eq:div-one-template},
we derive that for every query $q$ and every $r \in S_q$,
there exists exactly one plan $p = p^r \in \TPlans{q}$ 
such that $\vastar(y^r_p) = 1$.

Third, based on~\eqref{eq:div-atomic},
there exists an atomic configuration $Y^{p, r}$, $r \in S_q$, $p = p^r$
that corresponds to the assignments for $\va(x^r_{pa})$.

Finally, we prove that $p^r$ and $Y^{p, r}$, $r \in S_q$, 
correspond to 
the choice of plan $p$ and atomic configuration $Y$ that yields the minimum value
of $cost(q, \confcomponent{r})$, by 
using contradiction. 
Combining these results, we conclude that $BIPCost(\vastar) = \totalcost(\configuration, \routing)$.

Suppose that 
there exists a different choice $p^c \in \TPlans{q}$
and $Y^c \in \atom(\confcomponent{r})$, $r \in S_q$,
such that $cost(q, p^c, Y^c) < cost(q, p^r, Y^{p, r}$). 
Here, we use $cost(q, p, Y)$ denote the cost of $q$ 
using the template plan $p$ and the atomic configuration $Y$.

We can now derive an alternative assignment $\vacounter$ 
that is similar to $\vastar$ except the followings:
\begin{itemize}
	\item Variables corresponding to $p^r$ and $Y^{p, r}$ are assigned value $0$,  and
	\item $\va(y^r_p) = 1$ if $p = p^c$, $r \in S_q$, and 
	\item $\va(x^r_{pa}) = 1$, if $p = p^c$, $r \in S_q$ and $a = Y^{c}[i]$, $i \in [1,n]$. 	
\end{itemize}

We observe that  $\vacounter$
is a valid constraint-assignment for the formulated 
BIP. However, since $BIPcost(\vacounter) < BIPcost(\vastar)$, 
this contradicts our assumption about the optimality of $\vastar$.

\section{Factoring Failures}
\label{app:failure}

In this section, we present the full details of how \rita\ integrates failures into the BIP.

Under our assumption of using fast what-if optimization, the cost of a query $q$ in some replica $r$ can be expressed as $\cost(q,I_r) = \cost(p',A')$ for some choice of $p' \in \tplans(q)$ and an atomic configuration $A' \in \atom(I_r)$ We introduce the following additional variables. 

\begin{itemize}
	\item $t^{r,j}_q=1$ if and only if $q$ is routed to replica $r$ when $j$ fails, i.e., 
$\route{j}(q)=\{r \ | \ t^{r,j}_q=1\}$
	\item $x^{r,j}_{pa} = 1$ if and only if $q$ is routed to replica $r$ when $j$ fails, 
	$p=p'$ and $a \in A'$.
	\item $y^{r, j}_p = 1$ if and only if $q$ is routed to replica $r$ when $j$ fails, $p=p'$. 
\end{itemize}

We also need to add a new set of constraints, as given in Figure~\ref{fig:failure}.
These constraints are very similar to their counterparts  in Figure~\ref{fig:bip_basic}.
The correctness of the BIP is proven in the same way as presented in Appendix~\ref{app:div-equivalent}. 

\begin{figure}[t]
{\small
	  \begin{equation*} 
		\ftotalcost (\configuration, \routing, j) 
				 = \sum_{ q \in Q} \sum_{r \in [1,N] \wedge r \neq j} 
					 \frac{f(q)}{\max{m, N-1}} \costbip(q, r, j)			
		\end{equation*}					 
	
	 \begin{equation} 
	\costbip(q, r, j) = 
		\sum_{p \in \TPlans{q}}\beta_{p}y^{r, j}_{p} + 
		\sum_{\substack{p \in \TPlans{q} \\ a \in \Iset \cup 
			\{ \noindextiny{1} \} \cup \cdots \cup \{ \noindextiny{n} \}
			}}
		\gamma_{pa}x^{r, j}_{pa}, 		
		\substack{\forall r \in [1, N],  \\ \forall q \in Q \cup Q_{upd}}
 \end{equation}
 such that:	
	\begin{equation}
		\sum_{r \in [1, N]} t^{r, j}_q = \max\{N-1, m\}, \forall q \in Q
	\end{equation}
	\begin{equation}
		\sum_{p \in \TPlans{q}} y^{r, j}_{p} = t^{r, j}_q,\ \ \ \forall q \in Q \cup Q_{upd}		
	\end{equation}
  \begin{equation} 
  	s^{\rep}_a  \geq x^{r, j}_{pa},
		\ \ \ \forall q \in Q \cup Q_{upd}, 
			 p \in \TPlans{q}, \  a \in \Iset
  \end{equation}
 \begin{equation} 
     \sum_{a \in \Iset_i \cup \{ \noindextiny{i} \} }x^{r, j}_{pa} = y^{r, j}_{p}, 
		\ \ \substack{ \forall q \in Q \cup Q_{upd},  p \in \TPlans{q}, \\
						  i \in [1,n], \ T_i \mbox{ is referenced in q} 
					 }	 
 \end{equation} 
 
}
\caption{Augmented BIP to handle failures. \label{fig:failure}}
\end{figure}

\section{Bounding Load-skew}

\subsection{Additional Constraints for Exact Solution}
\label{app:exact-balance}

This section presents the set of constraints 
that \rita\ formulates in order to ensure the optimality of $\costbip(q,r)$
with the presence of bounding load-skew constraints.

\rita\ introduces a new cost formula $\costopt(q,r) = cost(q, I_r)$ for
$r \in [1,N]$.
The formula of $\costopt(q,r)$ is very similar to $cost(q,r)$;
the variables $\yopt^{r}_{p}$ (resp. $\xopt^{r}_{pa}$)
have the same meaning with $y^{r}_{p}$ (resp. $x^{r}_{pa}$).
The main difference is that
for $r \not \in \rcomponent{0}(q)$,
we have $\costbip(q, r) = 0$ whereas $\costopt(q,r) = cost(q, I_r) > 0$.
The atomic constraint in \eqref{eq:div-atomic-opt}
are somehow similar to the atomic constraints on $cost(q,r)$.
Note that in \eqref{eq:div-atomic-opt-first}, the constraint
requires exactly one template plan to be chosen to compute 
$\costopt(q, r)$ in order for this value corresponds
to the query execution cost of $q$ on replica $r$.

\begin{figure}[t] 
{\small
\begin{equation} \label{eq:div-cost-opt}
	\costopt(q, r) = 
		\sum_{p \in \TPlans{q}} 
			\beta_{p}\yopt^{r}_{p} 
		+ 
		\sum_{ \substack{
				p \in \TPlans{q} 
					\\ a \in \Iset \cup \{ \noindextiny{1} \} \cup \cdots \cup \{ \noindextiny{n} \} 
				} 
		     }	
		\gamma_{pa}\xopt^{r}_{pa}, 
		\substack{ \forall q \in Q \cup Q_{upd} \\ \forall r \in [1, N]}
 \end{equation}
 \begin{subequations} \label{eq:div-atomic-opt}
	\begin{align}
		\sum_{p \in \TPlans{q}} \yopt^{r}_{p} &= 1 \label{eq:div-atomic-opt-first} \\
		\sum_{
			a \in \Iset_i \cup \{ \noindextiny{i} \}
			}
		\xopt^{r}_{pa} &= \yopt^{r}_{p}, 
		\ \ 
		\substack{ 
			\forall p \in \TPlans{q} \\
			 \forall i \in [1,n]  \wedge \ T_i \mbox{ is referenced in q} 
			 }	
		 \label{eq:div-atomic-opt-second} 
	\end{align}
 \end{subequations} 
	\begin{equation}\label{eq:optimal-beta-opt}
		\costopt(q, r) \leq 
				\beta_{p} + \sum_{\substack{i \in [1,n] \\ a \in \Iset \cup 
			\{ \noindextiny{1} \} \cup \cdots \cup \{ \noindextiny{n} \} } } 													\gamma_{pa}u^{r}_{pa}, \ \  \forall p \in \TPlans{q}
	\end{equation}
	\begin{subequations}\label{eq:optimal-atomic-opt}
	\begin{align}
		\sum_{a \in \Iset_i \cup I_{\emptyset}} u^{r}_{pa} &= 1, 
		\ \ \substack{ \forall t \in [1,K_q], \\ \forall i \in [1,n] \ \wedge \ T_i 
			\mbox{ is referenced in q} 
			}		
		\label{eq:optimal-atomic-opt-1} \\
		u^{r}_{pa} & \leq s^{r}_a, \forall p \in \TPlans{q} \wedge a \in \Iset 
		\label{eq:optimal-atomic-opt-2} 
	\end{align}		
	\end{subequations}
	\begin{equation} \label{eq:optimal-gamma-opt}
        \sum_{b \in \Iset_i \cup I_{\emptyset} \ \wedge \ \gamma_{pa} \geq \gamma_{pb}}
        	u^r_{pb} \geq s^r_a, \ \ \forall p \in \TPlans{q}, i \in [1,n], a \in \Iset_i
	\end{equation}

}
\caption{Query-Optimal Constraints \label{imbalance-common}}
\end{figure}

To establish the optimal cost constraints,
we use the following alternative
way to compute $cost(q,X)$. 
For each internal plan cost $\beta_p$, $p \in \TPlans{q}$, 
we first derive a ``local'' optimal cost,
referred to as $C^{local}_t$, 
which is the smallest cost that can be obtained by ``plugging'' 
all possible atomic configurations $A \in \atom(X)$ into the slot of 
the template plan of $\beta_{p}$.
Essentially, $C^{local}_t = \beta_{p} + I^{local}_p$,
where $I^{local}_p$ is the smallest value of the total access cost 
using some atomic-configuration $A \in \atom(X)$ to plug into the template 
plan of $\beta_p$. 
To obtain $I^{local}_p$, for each slot in the internal plan of $\beta_{p}$, 
we enumerate all possible indexes in $X$ that can be ``plugged'' into, 
and find the one that yields the smallest access cost to sum up into $I^{local}_p$. 
Lastly, $cost_q(X)$ is then obtained 
as the smallest value among the derived $C^{local}_p$ with $p \in \TPlans{q}$.

The right hand-side of \eqref{eq:optimal-beta-opt} 
is the formula of $C^{local}_p$. Here, we introduce 
variables $u^{r}_{pa}$; where
$u^{r}_{pa} = 1$  iff the index $a$ is used at slot $i$ 
in the template plan $\beta_{p}$
to compute $C^{local}_p$. 	
For $C^{local}_p$ to correspond to some atomic configuration,
we impose the constraint in \eqref{eq:optimal-atomic-opt-1}.

Furthermore, an index $a$ can be used in $C^{opt}_p$ if and
only if $a$ is recommended at replica $r$ (constraint 
\eqref{eq:optimal-atomic-opt-2}).

The constraint \eqref{eq:optimal-gamma-opt} ensures that
the candidate index with the smallest 
access cost is selected to plug into each slot 
of $\beta_t$ in computing $I^{local}_t$.

\subsection{Greedy Approach}
\label{app:greedy}

This section presents our proposal of a greedy scheme 
that trade-offs the quality of the design for the efficiency. 

First, we derive 
an optimal design $(\configuration_{opt}, \routing_{opt})$
assuming there is no load imbalance constraint
and the probability of failure is $0$. 
We then compute
an approximation factor $\beta = \frac{\nodefactor - 1}{1 + (N - 1) \nodefactor}$.
and add
the following constraint into the BIP. 
\begin{equation} \label{eq:div-greedy}
\load(\configuration, \mapping, r) \leq \frac{(1 + \beta) \totalcost(\configuration_{opt}, \mapping_{opt})}{N}, 
\forall r \in [1,N]
\end{equation}

This constraint is an easy constraint, as its right handside is a constant.
We prove that if the BIP solver can find a solution
for the modified BIP,
the returned solution is a valid solution 
and has $\totalcost(\configuration, \routing)$ bounded
as the following theorem shows.

\begin{theorem} \label{theorem:node-factor}
The divergent design returned by the greedy solution
satisfies all constraints in \divergent\ problem 
and has \linebreak $\totalcost(\configuration, \routing) \leq 
(1 + \beta) \totalcost(\configuration_{opt}, \routing_{opt})$. \eop
\end{theorem}

\begin{proof}

We overload $\Iopt$ (resp. $\configuration$) to refer to the total cost
of the design $I_{opt}$ (resp. $\configuration$) as well.

The maximum load of a replica in $\configuration$
is $\frac{(1 + \beta) \Iopt}{N}$
(due to the constraint~\ref{eq:div-greedy}).
By summing up the load 
of all replicas in $\configuration$,
we obtain: $\configuration \leq (1 + \beta) \Iopt$.
Therefore, $\configuration$ differs from
$\Iopt$ by an approximation ratio $(1 + \beta)$. 
All remaining issue is to prove 
that $\configuration$ satisfies the load-imbalance constraint.

Without loss of generality, assume that $load(1, \configuration) \leq load(j, \configuration)$, $\forall j \in [2,N]$.

Since $\configuration$ is load-imbalance,
we can derive the followings:

\begin{subequations}\label{eq:prove-atomic}
			\begin{align}
				\configuration = \sum_{j \in [2,N] }
				\load(j, \configuration) + \load(1, \configuration)
				\\
				\frac{(1 + \beta) \Iopt}{N} 
				\geq \load(j, \configuration)
				\\
				\frac{(N - 1)} {N} (1 + \beta) \Iopt
				+ load(1, \configuration)
				\geq \configuration	 \\			
				\configuration \geq \Iopt
				\\
			load(1, \configuration) \geq \left( 1 -  \frac{(N - 1)} {N} (1 + \alpha) \right) \Iopt
			\end{align}
		\end{subequations}

The maximum load in $\configuration$
is $\frac{1}{N} (1 + \alpha) \Iopt$
and the minimum load
is $\left( 1 -  \frac{(N - 1)} {N} (1 + \alpha) \right) \Iopt$.
Therefore, the load-imbalance factor of $\configuration$
is $\frac{1 + \beta}{ 1 - (N - 1) \alpha}$. 
By replacing the value of $\beta$,
we obtain the load-imbalance factor $\nodefactor$. 
\end{proof}

Note that this greedy scheme 
does not encounter the aforementioned problem
with $\costbip(q, r)$ not to be equal to
$cost(q,I_r)$.
Informally, the reason is due to the fact that the right hand-side 
of the inequality constraint in~\eqref{eq:div-greedy}
is a constant.

\eat{
\section{Additional Intra-Replica Constraints}
\label{app:additional-local-constraints}

This section presents some additional intra-replica constraints
that can be easily incorporated into the framework of \divgbip

\subsection{Index Constraints}
This type of constraints specify conditions on the indexes
in the recommendation, such as conditions on their size (i.e., the space budget constraint),
contained columns, or their column-width. 

For instance, the DBA can specify that at most $2$
indexes containing more than $5$ columns
should be selected on the table $T_i$ and on a particular replica.
\divgbip\ first derives a subset $S_c \subset \Iset$
that contains all indexes with more than $5$ columns,
and sets the following constraints:

\begin{equation}
	\sum_{a \in S_c} s^r_a \leq  2,  \forall r \in [1, N]
\end{equation}

\subsection{Update Cost Constraints}

In replicated databases, when more replicas are deployed, the total update cost
usually increases, since we need to perform the update at more replicas.
Therefore, it will be useful if we can somehow 
bound the update cost that the replicated
databases need to perform in addition to space budget constraint.
This type of constraint
can also be easily incorporated into our BIP-based
framework. 
More specifically, let $B_{upd}$ denote 
the update bound of update cost. 
A constraint that bounds
the update cost on a divergent system with $N$
replicas can be integrated with one
additional constraint: $\updatecost < B_{upd}$.
Since $\updatecost$ is a linear expression, this constraint
still preserves the linear property of the formulated BIP. 
There are some ways that the DBAs can determine
the value of $B_{upd}$.
For instance, $B_{upd}$ can be the update cost
on a uniform design of a replicated database with less number of replicas.
As another example, \TODO{Find other alternatives.}
}

\subsection{Additional Experimental Results}
\label{app:vary-m}

This section presents the comparison 
between \rita, \divgdesign\
and \unif\
when we vary the routing multiplicity factor.
Figure~\ref{fig:vary-m} presents
one representative result
when we vary this factor on $\static$ workload
with $b = 0.5\times$ and $N = 3$.

As expected, when the value of $m$ increases,
the total cost of $\prita$ and $\pdd$
increase, since queries need to be sent to more places.
Note that the cost of $\punif$ remains the same,
as all replicas have the same index configuration under UNIF design.
Also, when $m = N$,
the total costs of $\pdd$ and $\punif$
are the same, since $\divgdesign$
needs to send every query to every replica,
and it uses the same black-box design advisor as \unif\ 
to compute the recommended index-set at each replica.

We observe that in all cases, 
\rita\ significantly outperforms \divgdesign\ and \unif. 
The reason can be again attributed to the fact that
\rita\ searches a considerably larger space of possible designs.

\begin{figure}[t]
        \begin{center}
                   \includegraphics[width=0.6\hsize]{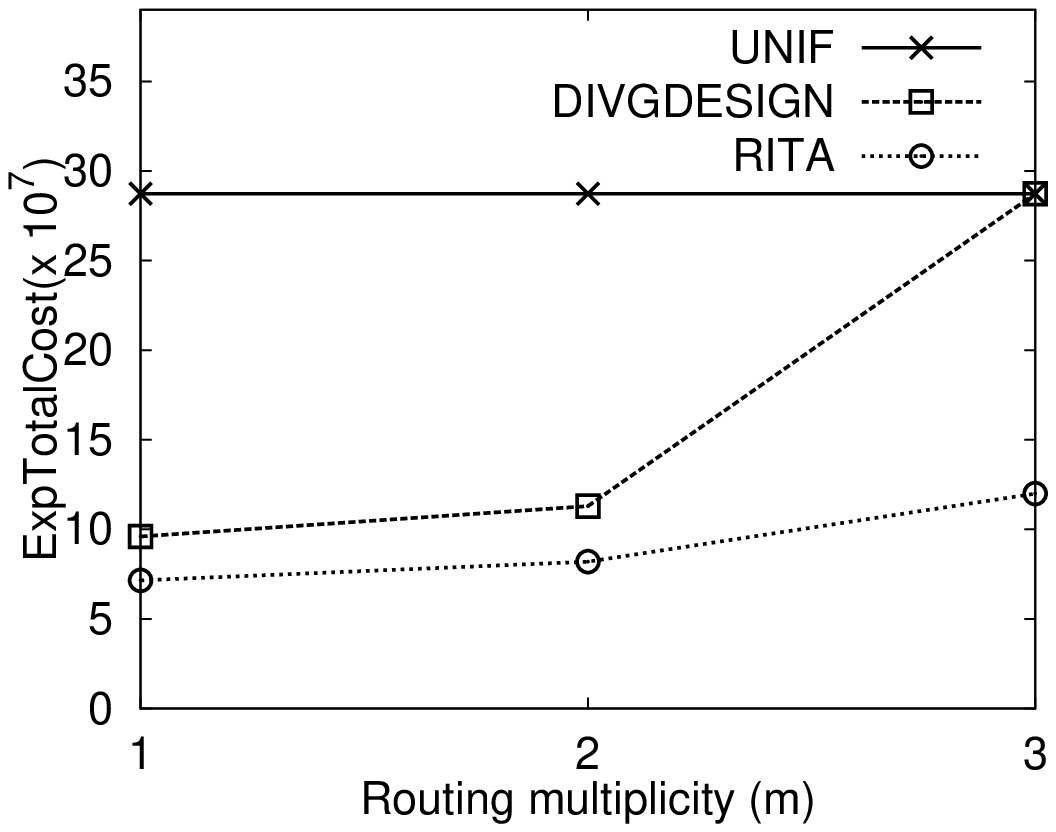}
                \caption{Varying the routing multiplicity factor on $\static$ workload}
\label{fig:vary-m}
        \end{center} 
\end{figure}